\documentclass[a4paper,twocolumn,11pt,unpublished]{quantumarticle}

\pdfoutput=1
\usepackage[utf8]{inputenc}
\usepackage[english]{babel}
\usepackage[table, dvipsnames]{xcolor}
\usepackage[T1]{fontenc}
\usepackage{amsmath}
\usepackage{amssymb}
\usepackage[allcolors=magenta]{hyperref}
\usepackage[capitalize]{cleveref}
\usepackage{tikz}
\usepackage{lipsum}
\usepackage{graphicx}
\usepackage{epsfig}
\usepackage{epstopdf}
\usepackage{amsthm}
\usepackage{enumitem}
\usepackage{dsfont}
\usepackage{bbold}
\usepackage{bm}
\usepackage{eurosym}
\usepackage{diagbox}
\usepackage[numbers,sort&compress]{natbib}
\usepackage[ruled,vlined]{algorithm2e}
\usepackage{mathtools}
\usepackage{mathabx} 
\usepackage{multirow}
\usepackage{mathtools}
\usepackage{physics}
\usepackage{hhline}
\usepackage{comment}
\usepackage{array}

\usepackage{tcolorbox}
\newtcolorbox{mybox}[2][]{
               = {yshift=-8pt},
  colback      = cyan!6!white,
  colframe     = cyan!1!black,
  halign       = flush left,
  fonttitle    = \bfseries\sffamily,
  colbacktitle = cyan!50!black,
  title        = #2,#1,
  }
\usepackage{ragged2e}

\newcommand{\be}{\begin{equation}}
\newcommand{\ee}{\end{equation}}
\newcommand{\ba}{\begin{aligned}}
\newcommand{\ea}{\end{aligned}}

\newcommand{\bc}{\begin{center}}
\newcommand{\ec}{\end{center}}
\newcommand{\beq}{\begin{equation}}
\newcommand{\eeq}{\end{equation}}
\newcommand{\beqq}{\begin{equation*}}
\newcommand{\eeqq}{\end{equation*}}
\newcommand{\beqa}{\begin{align}}
\newcommand{\eeqa}{\end{align}}
\newcommand{\barr}{\begin{array}}
\newcommand{\earr}{\end{array}}
\newcommand{\bi}{\begin{itemize}}
\newcommand{\ei}{\end{itemize}}

\newcommand{\C}{\mathbb{C}}

\newtheorem{lemma}{Lemma}

\newtheorem{theorem}{Theorem}

\newtheorem{definition}{Definition}

\begin{document}


\title{Extreme non-negative Wigner functions}


\author{Zacharie Van Herstraeten}
\orcid{0000-0003-1810-0942}
\affiliation{DIENS, \'Ecole Normale Sup\'erieure, PSL University, CNRS, INRIA, 45 rue d'Ulm, Paris 75005, France}
\email{zacharie.van-herstraeten@inria.fr}

\author{Jack Davis}
\affiliation{DIENS, \'Ecole Normale Sup\'erieure, PSL University, CNRS, INRIA, 45 rue d'Ulm, Paris 75005, France}
\email{jack.davis@inria.fr}

\author{Nuno C. Dias}
\affiliation{Grupo de Física Matemática, Departamento de Matemática \\Instituto Superior Técnico, Av. Rovisco Pais, 1049-001 Lisbon, Portugal}
\email{nunocdias1@gmail.com}

\author{João N. Prata}
\affiliation{Grupo de Física Matemática, Departamento de Matemática \\Instituto Superior Técnico, Av. Rovisco Pais, 1049-001 Lisbon, Portugal}
\affiliation{ISCTE - Instituto Universitário de Lisboa, Av. das
    Forças Armadas, 1649-026, Lisbon, Portugal}
\email{joao.prata@mail.teleapc.pt}

\author{Nicolas J. Cerf}
\affiliation{Centre for Quantum Information and Communication, \'{E}cole polytechnique de Bruxelles,  CP 165, \\Universit\'{e} libre de Bruxelles, 1050 Brussels, Belgium}
\email{nicolas.cerf@ulb.be}

\author{Ulysse Chabaud}
\orcid{0000-0003-0135-9819}
\affiliation{DIENS, \'Ecole Normale Sup\'erieure, PSL University, CNRS, INRIA, 45 rue d'Ulm, Paris 75005, France}
\email{ulysse.chabaud@inria.fr}



\begin{abstract}
    Providing an operational characterization of the Wigner-positive states (WPS), i.e.,~the set of quantum states with non-negative Wigner function, is a longstanding open problem. For pure states, the only WPS are Gaussian states, but the situation is considerably more subtle for mixed states.
    Here, we approach the problem using convex geometry, reducing
    the question
    to the characterization of the extreme points of the set of WPS.
    We give a constructive method to generate a large class of such extreme WPS, which combines the following steps: (i)~we characterize the phase-invariant extreme points of the superset of  Wigner-positive quasi-states (WPQS); (ii)~we introduce a new quantum map, named \textit{Vertigo} map, which maps extreme WPQS to extreme WPS while preserving phase invariance; (iii)~we identify families of extremality-preserving maps and use them to obtain extreme WPS while relaxing phase invariance. Our construction generates all extreme WPS of low dimension, starting from a specific kind of WPS known as beam-splitter states.
    Our results build upon new mathematical properties of the set of WPS derived in a companion paper \cite{mathpaper} and unveil the remarkable structure of mixed states with non-negative Wigner functions.
\end{abstract}

\maketitle

\section{Introduction}

The negativity of phase-space representations of quantum states is a peculiar property of quantum mechanics \cite{rundle2021overview}. In particular, the negativity of the Wigner function \cite{Wigner1932} is a marker of non-classicality connected to quantum contextuality \cite{spekkens2008negativity,booth2021,haferkamp2021equivalence} and has emerged as a necessary resource for quantum computational advantages \cite{kenfack2004negativity,Mari2012,albarelli2018resource, takagi2018convex}.
Consequently, the set denoted $\mathcal{D}_{+}$ of quantum states with non-negative Wigner function---or Wigner-positive states (WPS) for short---forms the complement to the set of resourceful Wigner-negative states. As such, it is thus an important subject of investigation for both a fundamental understanding and practical applications of quantum information science. In the infinite-dimensional setting, however, a complete characterization is only known in the case of pure states, for which Hudson \cite{hudson1974wigner}---and later Soto \& Claverie \cite{soto1983wigner} in the multimode case---have shown that pure WPS are in one--to-one correspondence with pure Gaussian states. Despite over fifty years of research \cite{hudson1974wigner,soto1983wigner,jagannathan1987dynamical,garcia1988nonnegative,brocker_mixed_1995, gross2006hudson,dias2008narcowich,mandilara2009extending,Chabaud2021witnessing,Van_Herstraeten2021-nj}, the explicit description of $\mathcal{D}_{+}$ remains limited.

An analogous problem exists in certain finite-dimensional quantum systems (i.e., qudits) or spin systems. When appropriately defined, the set of WPS is known to be free in the resource-theoretic context of magic state distillation, a critical component in a leading model of fault-tolerant quantum computation \cite{veitch2012negative}.  The complementing set of Wigner-negative states enjoys a similar relationship to contextuality \cite{howard2014contextuality, delfosse2017equivalence} and is conjectured to coincide with the set of distillable magic states \cite{veitch2014resource}, though explicitly finding the set of error-correcting distillation codes that would validate this conjecture is difficult \cite{Campbell_prime_MSD_Reed_Muller2012PRX, Dawkins_Howard_2015, prakash_bound2020, prakash2024qutrit_search}.

In any situation where a Wigner function is well-defined, the determination of $\mathcal{D}_{+}$ can be abstractly summarized by taking the convex hull of the Wigner frame operators
(i.e., operators whose associated Weyl symbol is maximally localized, in the sense of a Dirac or Kronecker delta)
and intersecting it with the cone of positive semi-definite operators.  Crucially, the result is always a convex set.  In the finite-dimensional setting, $\mathcal{D}_{+}$ is compact and therefore subject to the Krein--Milman theorem, which states that a compact convex set is completely determined by its extreme points: it is the closed convex hull of these points.  In the infinite-dimensional setting (to which we restrict our attention in the remainder of this paper), $\mathcal{D}_{+}$ is non-compact and therefore beyond the scope of the Krein--Milman theorem.  However, we prove in a companion paper~\cite{mathpaper} that a version of the Krein--Milman theorem nevertheless holds specifically for $\mathcal{D}_{+}$, i.e., the non-compact set of 
WPS (see \cref{th:Krein-Milman}).  Consequently, extreme WPS become the essential building blocks of $\mathcal{D}_{+}$.

\begin{figure}[t]
    \centering
    \includegraphics[width=0.8\linewidth]{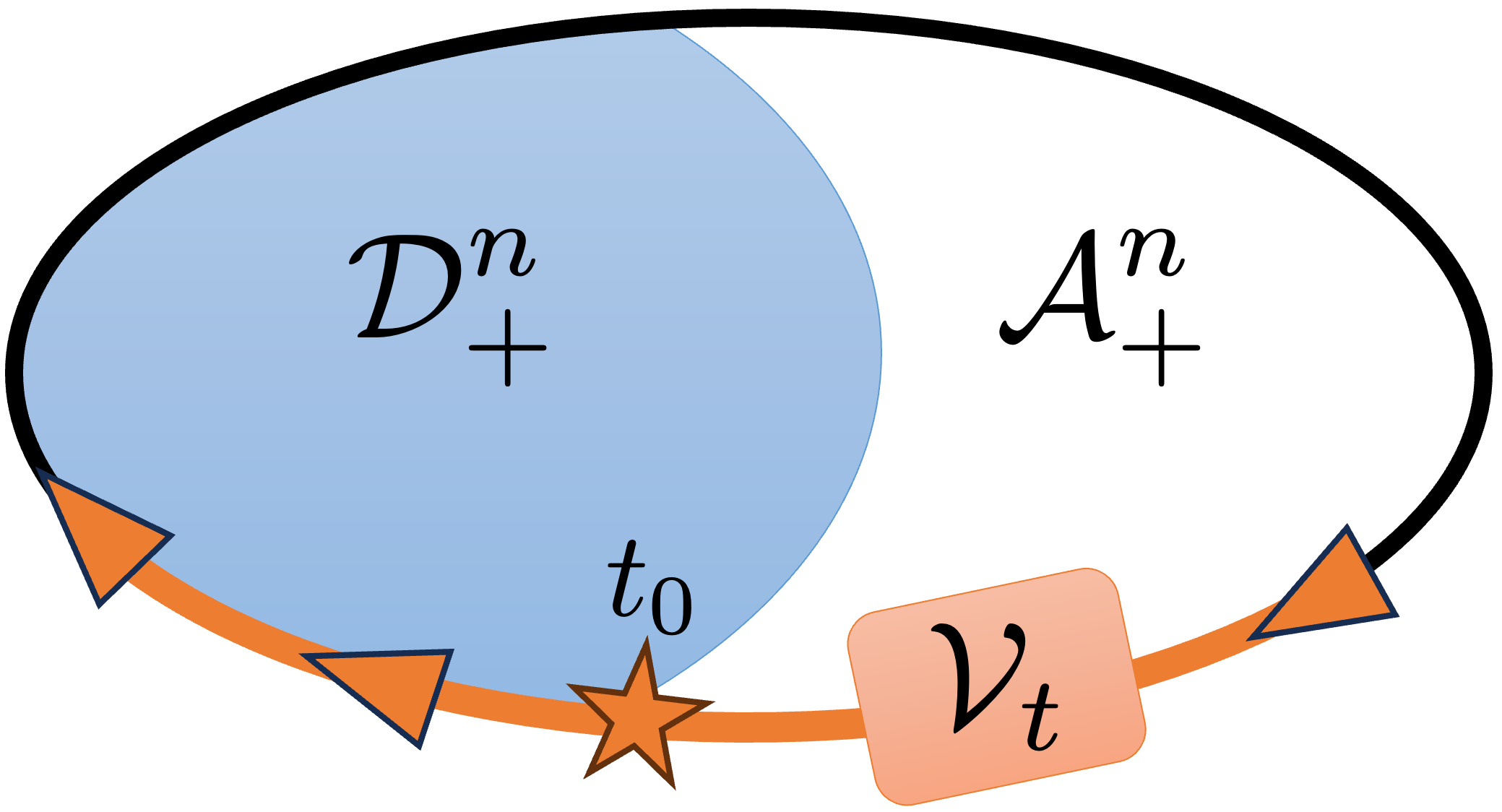}
    \caption{\small \textbf{Mapping extreme WPQS onto extreme WPS.}
        Within the cut of operators with Fock-bounded support (i.e., a bounded particle number $n$), the set of Wigner-positive states $\mathcal{D}^{n}_{+}$ (blue shaded region) appears as a convex subset of the Wigner-positive quasi-states $\mathcal{A}^{n}_{+}$. The Vertigo map $\mathcal{V}_{t}$ (orange arrows) sends extreme points of $\mathcal{A}^{n}_{+}$ to extreme points of $\mathcal{D}^{n}_{+}$ as soon as the parameter $t$ exceeds a threshold value $t_{0}>0$ (with $t_{0}$ generally depending on the initial point). The fixed points of the Vertigo map $\mathcal{V}_{t}$ (for $t\to\infty$) are binomial states. The last step of the construction of WPS (pictured in \cref{fig:extreme_state_generation}) relaxes Fock-boundedness while preserving extremality.
    }
    \label{fig:intro_vertigo_simple}
\end{figure}

In this work, we make substantial progress on the characterization of $\mathcal{D}_{+}$ by uncovering properties of its extreme points, using tools from convex geometry and polynomial analysis. Our central contribution is a constructive method, sketched in Fig.~\ref{fig:intro_vertigo_simple}, which generates a large class of extreme WPS, providing more insights into the structure of $\mathcal{D}_{+}$.
The starting point of our construction is to relax the positive-definiteness condition on quantum operators and hence consider \textit{quasi-states}. Unlike quantum states, which are associated with positive-definite trace-one Hermitian operators, quasi-states are defined as trace-one Hermitian operators that must not be positive definite. We thus focus on the superset denoted $\mathcal A_+$ of quasi-states with non-negative Wigner function---or Wigner-positive quasi-states (WPQS) for short---and characterize the extreme quasi-states within some tractable subsets of $\mathcal A_+$ (see \cref{th:extrAo+n}).  More precisely, we restrict to subsets of $\mathcal A_+$ that correspond to Fock-bounded operators, with bounded particle number $n$, denoted as $\mathcal A_+^n$. These may be positive semi-definite (for states) or not (for quasi-states). The reason for considering Fock-bounded operators is that the sign of their corresponding Wigner functions is governed by a polynomial and can be conveniently analyzed using tools from analysis. By further restricting to the subset denoted $\mathcal A_\oplus^n$ of Fock-bounded operators that are phase-invariant (i.e., whose Wigner function is rotation-invariant), 
the problem becomes fully tractable and we are able to characterize the extreme WPQS.

Our second step is to move from such extreme WPQS in $\mathcal A_\oplus^n$ toward extreme WPS in $\mathcal{D}_{+}$.
Here, a key ingredient is to introduce a parametrized quantum probabilistic map $\mathcal V_t$---named \textit{Vertigo} map---which preserves extremality (as well as phase-invariance). In particular, we show that the Vertigo map can be expressed as a combination of a pure loss channel and a noiseless linear amplifier (see~\cref{fig:vertigo_physical_implementation}), and sends extreme WPQS to extreme WPS as the parameter $t$ increases (see \cref{fig:intro_vertigo_simple} and \cref{th:vertigoWPQStoWPS}). We also characterize the fixed points of the Vertigo map in terms of phase-invariant \textit{binomial states}, which play a central role in the characterization of extreme WPS.

Finally, as a third step, we identify families of extremality-preserving operators which map extreme WPS to extreme WPS while breaking phase invariance. Among these, we consider a displacement-like operator $\bar D$, which preserves the Fock support, as well as Gaussian unitaries, which break the Fock boundedness. Applying these operators following the Vertigo map provides us with a method to cover a large class of (provably all, in low dimension) extreme WPS.

\begin{figure}
    \centering
    \includegraphics[width=\linewidth]{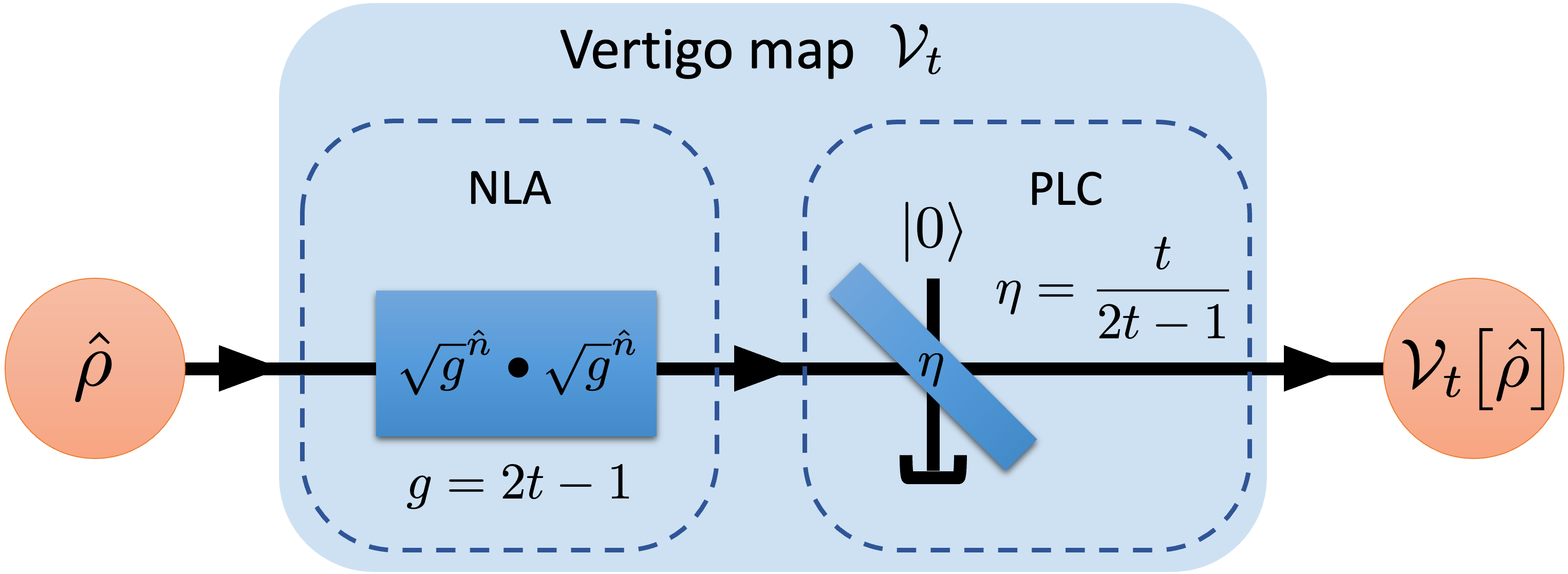}
    \caption{\small
        \textbf{The Vertigo map.} Physical implementation of the Vertigo map $\mathcal V_t$ for $t\geq 1$, as the concatenation of a noiseless linear amplifier (NLA) $\mathcal{N}$ and a pure loss channel (PLC) $\mathcal{E}$.
        We have $\mathcal{V}_{t}=\mathcal{E}_{\eta}\circ\mathcal{N}_{g}$ with $g=2t-1$ and $\eta=t/(2t-1)$.
    }
    \label{fig:vertigo_physical_implementation}
\end{figure}

Having introduced a general construction for generating extreme WPS, we then focus on a specific family of WPS states known as \textit{beam-splitter states} \cite{Becker2021-pw, Van_Herstraeten2021-nj,Van-Herstraeten2024-lc}, see \cref{fig:bs-state}. These states are produced by interfering two single-mode states at a balanced beam-splitter and tracing out one of the two output modes, and they include the above phase-invariant binomial states
as a special case (obtained when some Fock state interferes with the vacuum state at the balanced beam-splitter).
We prove that many beam-splitter states are, in fact, extreme WPS (see Theorem~\ref{theorem:beam-splitter-states-are-extreme}), and that the WPS that are the closest to Fock states are beam-splitter states (see Lemma~\ref{lemma:bs_states_fidelity_fock_states}). Notably, in low dimensions, we show that trajectories of beam-splitter states under the Vertigo map produce \textit{all} extreme WPS.



The rest of the paper is structured as follows: we give preliminary material on phase-space representations, beam-splitter states, and convex geometry in \cref{sec:preli}; in \cref{sec:extreme_wigpos_quasistates}, we analyse extreme WPQS, focusing on the subsets of phase-invariant Fock-bounded  WPQS; we then introduce the Vertigo map in \cref{sec:vertigo} and study its properties; in \cref{sec:generating}, we combine these results together with the use of extremality-preserving operators  in order to generate large classes of extreme WPS (we exhibit the special role of beam-splitter states in this context); finally, we discuss the consequences of our results in \cref{sec:concl} and raise a few open problems.




\section{Preliminaries}
\label{sec:preli}

We consider the infinite-dimensional Hilbert space $\mathcal{H}\simeq L^{2}(\mathbb{R})$ and denote by $\mathcal{B}_1(\mathcal{H})$ (or simply $\mathcal B_1$) the set of trace-class operators over $\mathcal{H}$.
We denote by $\mathcal{D}\subset\mathcal{B}_1$ the set of Hermitian positive semi-definite (PSD) operators of unit trace and refer to operators in $\mathcal{D}$ as \textit{states}.
We define $\mathcal{A}$ as the affine hull of $\mathcal{D}$, i.e.,~$\mathcal{A}\subset\mathcal{B}_1$ is the set of Hermitian operators of unit trace; we refer to operators in $\mathcal{A}$ as \textit{quasi-states}.
We have the inclusion $\mathcal{D}\subset\mathcal{A}\subset\mathcal{B}_1$.

\begin{figure}
    \centering
    \includegraphics[width=0.75\linewidth]{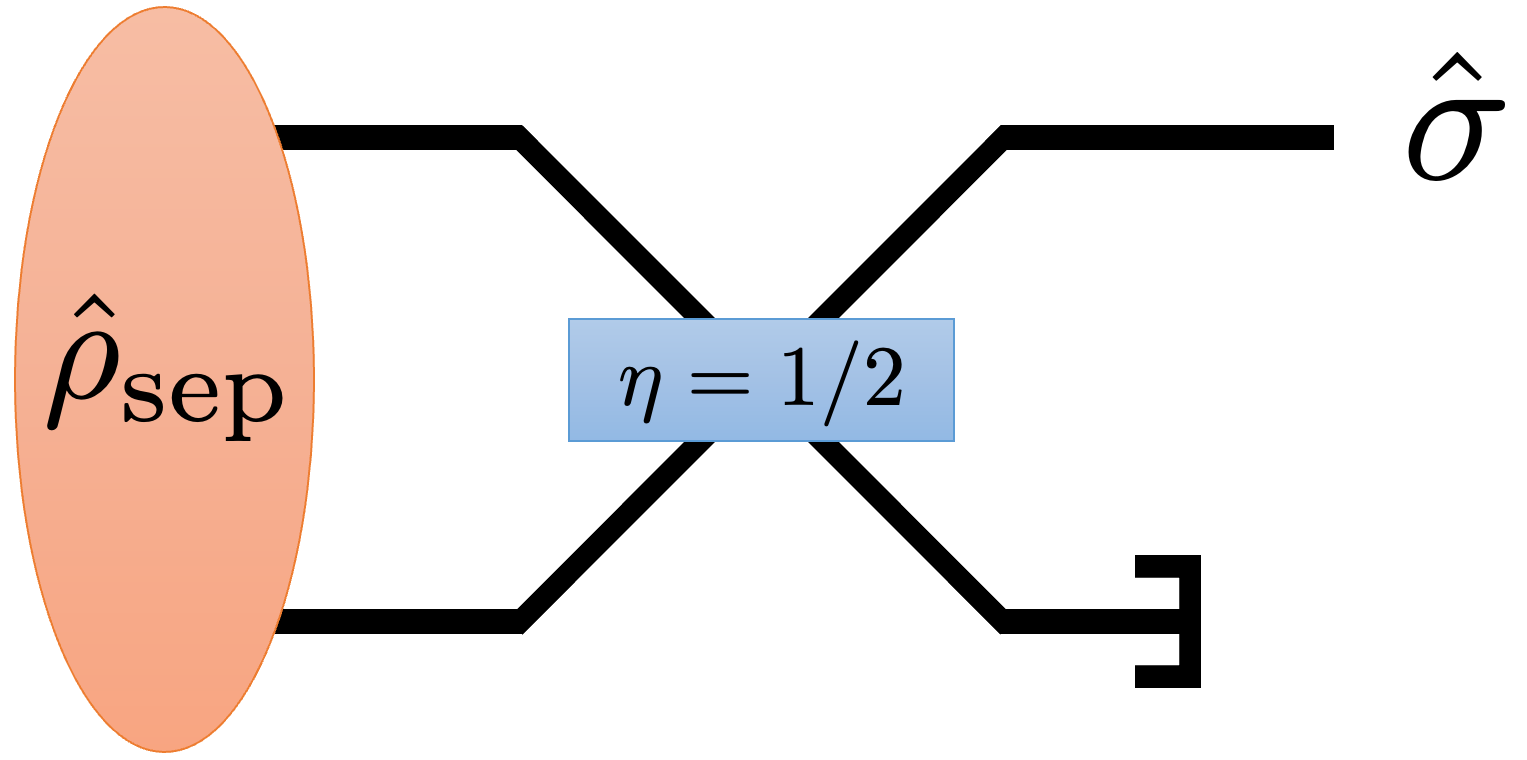}
    \caption{\small
    \textbf{Beam-splitter state.}
    A balanced beam-splitter ($\eta=1/2$) acts on a separable state $\hat{\rho}_{\mathrm{sep}}$.  After tracing over the second output mode, the beam-splitter state $\hat{\sigma}$ is obtained.
    Beam-splitter states are always Wigner positive.
    }
    \label{fig:bs-state}
\end{figure}

We recall the expressions of the annihilation operator $\hat{a}=(\hat{x}+i\hat{p})/\sqrt{2}$ and photon-number operator $\hat{n}=\hat{a}^{\dagger}\hat{a}$, the displacement operators $\hat{D}(\alpha)=\exp(\alpha\hat{a}^{\dagger}-\alpha^\ast\hat{a})$ for all $\alpha\in\C$, and the parity operator $\hat{\Pi}=(-1)^{\hat{n}}$.
We define the displaced parity operators as $\hat{\Pi}(\alpha)=\hat{D}(\alpha)\hat{\Pi}\hat{D}^{\dagger}(\alpha)$ for all $\alpha\in\C$.
We give a special attention to the Fock (or photon-number) basis $\lbrace\ket{n}\rbrace_{n\in\mathbb{N}}$ which comprises the eigenstates of the photon-number operator so that $\hat{n}\ket{n}=n\ket{n}$, for all $n\in\mathbb{N}$.
Coherent states are defined as $\ket{\alpha}=\hat{D}(\alpha)\ket{0}$.
We define $\mathcal{H}^{n}=\mathrm{span}\big(\lbrace\ket{k}\rbrace_{k\in[\![0,n]\!]}\big)$ as the finite-dimensional Hilbert space spanned by the first $n+1$ Fock states.

We introduce a few notations summarized in Table~\ref{table:notations_convex_sets}: for any set of unit-trace Hermitian operators ($\mathcal{D}$ or $\mathcal{A}$), we use the superscript ‘‘$n$'' to denote its subset of elements with no support beyond $n$ photons; we use the subscript ‘‘$+$'' to denote its subset of Wigner-non-negative elements, and the subscript ‘‘$\odot$'' to denote its subset of phase-invariant elements; we use the compact notation ‘‘$\oplus$'' to denote its subset of elements which are both Wigner-non-negative and phase-invariant.
For instance, $\mathcal{D}^{n}_{\oplus}$ is the set of phase-invariant WPS with no support beyond $n$ photons.

\begin{table}[t]
    \centering
    \resizebox{\linewidth}{!}{%
        \renewcommand{\arraystretch}{1.6}
        \begin{tabular}{|l|c|c|}
            \hline
            \textbf{Subset}                                                   & \textbf{Notation}                                           & \textbf{Definition}                                                                            \\
            \hhline{|=|=|=|}
            Fock-bounded                                                      & {\Large$\mathcal{A}^{n}$}                                   & {\large$\big\lbrace\hat{A}\in\mathcal{A}:\hat{P}_{n}\hat{A}\hat{P}_{n}=\hat{A}\big\rbrace$}    \\ \hline
            Wigner-positive                                                   & {\Large$\mathcal{A}_{+}$}                                   & {\large$\big\lbrace\hat{A}\in\mathcal{A}:W_{\hat{A}}(\alpha)\geq 0\ \forall\alpha\big\rbrace$} \\ \hline
            Phase-invariant                                                   & {\Large$\mathcal{A}_{\odot}$}                               & {\large$\big\lbrace\hat{A}\in\mathcal{A}:\mathcal{R}[\hat{A}]=\hat{A}\big\rbrace$}             \\ \hline
            \multirow{2}{2.7cm}{Wigner-positive phase-invariant}              & \multirow{2}{1cm}{\Large\centering$\mathcal{A}_{\oplus}$}   & \multirow{2}{3.8cm}{\Large\centering$\mathcal{A}_+\cap\mathcal{A}_{\odot}$}                    \\ & & \\\hline
            \multirow{2}{2.7cm}{Fock-bounded Wigner-positive phase-invariant} & \multirow{2}{1cm}{\Large\centering$\mathcal{A}^n_{\oplus}$} & \multirow{2}{3.8cm}{\Large\centering$\mathcal{A}^n\cap\mathcal{A}_{\oplus}$}                   \\ & & \\\hline
        \end{tabular}%
    }
    \caption{\small\textbf{Summary of notations.}
        These are notations for the different subsets of the set of quasi-states $\mathcal{A}$; for subsets of the set of states, replace $\mathcal{A}$ with $\mathcal{D}$.
        $\hat{P}_{n}\coloneq\sum_{k=0}^{n}\ket{k}\!\bra{k}$ is the projector on the first $n+1$ Fock states and $\mathcal{R}[\hat{\rho}]\coloneq\sum_{k=0}^{\infty}\bra{k}\hat{\rho}\ket{k}\ket{k}\!\bra{k}$ is the completely dephasing channel.
    }
    \label{table:notations_convex_sets}
\end{table}

\subsection{The Wigner function}

We consider a trace-class operator $\hat{A}\in\mathcal{B}_1$.
The \textit{Wigner function} of $\hat{A}$ is defined as
\begin{equation}
    W_{\hat{A}}(\alpha)
    =
    \frac{2}{\pi}
    \
    \mathrm{Tr}\left[
        \hat{A}\
        \hat{\Pi}(\alpha)
        \right].
    \label{eq:weyl_transform}
\end{equation}
It satisfies the inverse relation
\begin{equation}
    \hat{A}
    =
    2\int
    W_{\hat{A}}(\alpha)
    \;
    \hat{\Pi}(\alpha)
    \
    \mathrm{d}^2\alpha,
    \label{eq:weyl_transform_inverse}
\end{equation}
where $\mathrm{d}^2\alpha=\mathrm d\Re(\alpha)\mathrm d\Im(\alpha)$.
The Wigner function of the Fock transition operators $\ketbra{m}{n}$ are denoted $W_{\ket{m}\bra{n}}$ and are given in Appendix~\ref{apd:vertigo_in_state_space} in terms of Laguerre polynomials. Given an operator $\hat A \in \mathcal B_1$, the Wigner function of the adjoint $\hat A^\dagger$ is the complex conjugate of the Wigner function of $\hat A$.  This implies that Hermitian operators have real-valued Wigner functions.  We say that the Hermitian operator $\hat{A}$ is Wigner-non-negative if its Wigner function is non-negative, i.e.~$W_{\hat{A}}(\alpha)\geq 0$ $\forall\alpha$.

Given two operators $\hat{A},\hat{B}\in\mathcal{B}_1$, the overlap formula relates their Hilbert-Schmidt inner product to the $L^2$ inner product of their Wigner functions:
\begin{align}
    \mathrm{Tr}
    \big[
        \hat{A}\hat{B}^\dagger
        \big]
    =
    \pi
    \int
    W_{\hat{A}}(\alpha)
    W^*_{\hat{B}}(\alpha)
    \mathrm{d}^2\alpha.
\end{align}

\subsection{The Husimi function}

The Husimi function of an operator is defined as its overlap with a coherent state, i.e. the displaced vacuum state $\ket{\alpha}=\hat{D}(\alpha)\ket{0}$:
\begin{align}
    Q_{\hat{A}}(\alpha)
    =
    \frac{1}{\pi}\
    \bra{\alpha}
    \hat{A}
    \ket{\alpha},
\end{align}
for all $\alpha\in\C$.

The Husimi function can be obtained from the Wigner function through application of an additive Gaussian noise channel, which may be alternatively expressed as
\begin{align}\label{eq:WtoQ}
    Q_{\hat{A}}(\alpha)
    =
    \frac{1}{2}\
    W_{\mathcal{E}_{1/2}[\hat{A}]}
    \left(\frac{\alpha}{\sqrt{2}}\right),
\end{align}
where $\mathcal{E}_{1/2}$ is the pure-loss channel (PLC) with parameter $\eta=1/2$, defined as
\begin{align}
    \mathcal{E}_{\eta}
    \big[
        \hat{\rho}
        \big]
     & =
    \sum\limits_{k=0}^{\infty}
    \frac{(1-\eta)^k}{k!}
    \sqrt{\eta}^{\,\hat{n}}
    \ \hat{a}^k
    \ \hat{\rho}\
    \hat{a}^{\dagger k}\
    \sqrt{\eta}^{\,\hat{n}}.
    \label{eq:def_plc}
\end{align}
Observe that Eq.~\eqref{eq:def_plc} defines a linear map for any $\eta\in\mathbb{R}$.
However it is a proper quantum channel only in the regime $\eta\in[0,1]$.
In that case, the PLC can be implemented by mixing the input with the vacuum on a beam-splitter and tracing out one of the output modes.  Recall that the unitary operator of a beam-splitter is $\hat{U}_{\eta}=\mathrm{exp}\big(\theta(\hat{a}^\dagger\hat{b}-\hat{a}\hat{b}^{\dagger})\big)$ where $\eta=\cos^2\theta$; then,
\be
\mathcal{E}_{\eta}[\hat{\rho}]=\mathrm{Tr}_2\Big[\hat{U}_{\eta}\big(\hat{\rho}\otimes\ket{0}\!\bra{0}\big)\hat{U}^{\dagger}_{\eta}\Big].
\ee

The set of so-called \textit{beam-splitter states} form a notable family of WPS \cite{Becker2021-pw, Van_Herstraeten2021-nj, Van-Herstraeten2024-lc} and are based on a similar construction (see \cref{fig:bs-state}):

\begin{definition}[Beam-splitter state]
    A beam-splitter state $\hat{\sigma}$ is the single-mode output of a balanced beam-splitter acting on a separable state $\hat{\rho}_{\mathrm{sep}}$, i.e.~$\hat{\sigma}=\mathrm{Tr}_2\big[\hat{U}_{1/2}\,\hat{\rho}_{\mathrm{sep}}\,\hat{U}^{\dagger}_{1/2}\big]$.
\end{definition}

We will give particular attention to the subset of beam-splitter states built from the tensor product of two pure states, denoted as $\hat{\sigma}(\psi,\varphi)$. Moreover, we write these states as $\hat{\sigma}(m,n)$ when $\ket{\psi}=\ket{m}$ and $\ket{\varphi}=\ket{n}$ are Fock states. Furthermore, in the special case $n=0$, we denote states $\hat{\sigma}(m,0)$ as binomial states.

\subsection{Convex geometry}

We now take advantage of the convex structure of the set $\mathcal D_+$ and introduce the relevant definitions and results from convex geometry.

\begin{definition}[Extreme point]
    Given a convex set $\mathcal{C}$, the point $a\in\mathcal{C}$ is an extreme point of $\mathcal{C}$ if and only if $\forall x,y\in\mathcal{C}:a=\frac12(x+y)\Leftrightarrow x=y=a$.
\end{definition}

\noindent
Equivalently, $a$ is an extreme point of $\mathcal{C}$ if and only if it does not lie
in the interior of any line segment contained in $\mathcal{C}$.
We denote the set of extreme points of a convex set $\mathcal{C}$ as $\mathrm{Extr}(\mathcal{C})$.

As mentioned in the introduction, the relevance of extreme points in characterizing convex sets comes from the Krein--Milman theorem, which ensures that a compact and convex subset of a normed vector space is equal to the closed convex hull (denoted $\overline{\mathrm {Conv}}$) of its extreme points, i.e.,~the closure of the set of all convex combinations of extreme points.  The convex set $\mathcal D_+$ is not compact, but we show in a companion paper~\cite{mathpaper} that the theorem still holds for specifically $\mathcal D_+$:

\begin{theorem}[Krein--Milman theorem for WPS, \textit{informal statement of Theorem~40 in}~\cite{mathpaper}]\label{th:Krein-Milman}
    The set of WPS is equal to the closed convex hull of its extreme points:
    \be
    \mathcal D_+=\overline{\mbox{\rm Conv}}(\mathrm{Extr}(\mathcal D_+)),
    \ee
    where the closure is with respect to the trace norm.
\end{theorem}

\noindent This result implies that it is sufficient to characterize the extreme points of $\mathcal D_+$ to obtain an operational description of $\mathcal D_+$.

Let us now state two simple yet useful properties.
\begin{lemma}
    Let $\mathcal{C},\mathcal{C}'$ be general convex sets such that $\mathcal{C}'\subseteq\mathcal{C}$.
    Then, $\mathrm{Extr}(\mathcal{C})\cap\mathcal{C}'\subseteq\mathrm{Extr}(\mathcal{C}')$.
    \label{lemma:extreme_superset}
\end{lemma}

\noindent In simple words, considering the extreme points of $\mathcal{C}$ that belong to subset $\mathcal{C}'$ does not make them loose extremality within subset $\mathcal{C}'$.

\begin{lemma}
    Let $\mathcal{C}\subseteq\mathcal{D}$ be a convex subset of the convex set of quantum states $\mathcal D$ and $\hat{P}:\mathcal{D}\to\mathcal{D}$ a projector.
    Let $\mathcal{C}'=\hat{P}\mathcal{D}\hat{P}\cap\mathcal{C}$.
    Then, $\mathrm{Extr(\mathcal{C'})}\subseteq\mathrm{Extr}(\mathcal{C})$.
    \label{lemma:extreme_subset_projector}
\end{lemma}

\noindent In simple words, considering the extreme states of the restricted subset of $\mathcal{C}$ that belongs to the image of projector $P$ (the eigenspace of $P$ with eigenvalue one) does not make them loose extremality within the entire set $\mathcal{C}$.

We provide quick proofs of the two Lemmas in~\cref{app:proofs} (see also~\cite{mathpaper}).
These results motivate the consideration of subsets and supersets of $\mathcal D_+$ to identify some of its extreme points.
We will be particularly interested in Fock-bounded operators, i.e.,~operators with bounded support in the Fock basis, including both states and quasi-states.

Using the notations in Table~\ref{table:notations_convex_sets}, a notable consequence of Lemma~\ref{lemma:extreme_subset_projector}, for $\mathcal C=\mathcal{D}_{+}$, $\hat{P}=\sum_{k=0}^{n}\ket{k}\!\bra{k}$, and $\mathcal C'=\mathcal{D}_{+}^n$, is the inclusion:
\begin{align}
    \mathrm{Extr}(\mathcal{D}^{n}_{+})
    \subset
    \mathrm{Extr}(\mathcal{D}_{+}).
    \label{eq:fock-bounded_extreme_states_are_extreme}
\end{align}
This relation allows us to identify extreme points of $\mathcal{D}_{+}$ by restricting to the subset $\mathcal{D}^{n}_{+}$.
Importantly, Eq.~\eqref{eq:fock-bounded_extreme_states_are_extreme} does not generalize to quasi-states, i.e.,  $\mathrm{Extr}(\mathcal{A}^{n}_{+})\not\subset\mathrm{Extr}(\mathcal{A}_{+})$, as they lack positive semi-definiteness, which is a key ingredient to the proof of Lemma~\ref{lemma:extreme_subset_projector}.
Further, they do not form a convex subset of $\mathcal D$ so that $\mathrm{Extr}(\mathcal{A}^{n}_{+})\not\subset\mathrm{Extr}(\mathcal{D}_{+})$.
However, we can use Lemma~\ref{lemma:extreme_superset}, for $\mathcal C=\mathcal{A}^{n}_{+}$ and $\mathcal C'=\mathcal D_+^n$, to show that
\begin{align}
    \mathrm{Extr}(\mathcal{A}^{n}_{+})\cap\mathcal D_+^n\subset\mathrm{Extr}(\mathcal{D}_{+}^n).
    \label{eq:extreme_quasi-states_that_are_states_are_extreme_bounded}
\end{align}
Combined with Eq.~\eqref{eq:fock-bounded_extreme_states_are_extreme}, this implies
\begin{align}
    \mathrm{Extr}(\mathcal{A}^{n}_{+})\cap\mathcal{D}\subset\mathrm{Extr}(\mathcal{D}_{+}),
    \label{eq:extreme_quasi-states_that_are_states_are_extreme}
\end{align}
where we used $\mathrm{Extr}(\mathcal{A}^{n}_{+})\cap\mathcal{D}=\mathrm{Extr}(\mathcal{A}^{n}_{+})\cap\mathcal{D}_+^n$.
Eq.~\eqref{eq:extreme_quasi-states_that_are_states_are_extreme} summarizes the approach that we use to generate extreme WPS. We start from the extreme WPQS that are Fock-bounded and then keep those that are also states (positive semidefinite operators): according to Eq.~\eqref{eq:extreme_quasi-states_that_are_states_are_extreme}, these are guaranteed to be extreme WPS. Thus, the two main underlying questions can be stated as:
\begin{enumerate}
    \item[(Q1)] \emph{What is the set of extreme Fock-bounded quasi-states?}
          \\We address this question in~\cref{sec:extreme_wigpos_quasistates}.
    \item[(Q2)]
          \emph{When are such extreme quasi-states also extreme states?}
          \\We address this question in~\cref{sec:vertigo}.
\end{enumerate}

\noindent Providing answers to questions (Q1) and (Q2) allows us to generate a large class of extreme WPS via Eq.~\eqref{eq:extreme_quasi-states_that_are_states_are_extreme}, as we detail in~\cref{sec:generating}.

\section{Extreme Wigner-positive quasi-states (\texorpdfstring{$\mathcal A_+$}{})}
\label{sec:extreme_wigpos_quasistates}

We consider here the set of WPQS, $\mathcal{A}_{+}$.
Note that we do not expect the set $\mathcal{A}_{+}$ to possess any extreme points.
Indeed, any non-negative Wigner function can be written as a convex mixture of Dirac deltas, which correspond to the operators $2\hat{\Pi}(\alpha)$.
Although these operators are not quasi-states (they are not trace-class), they can be approximated arbitrarily well by quasi-states.
This suggests that $\mathcal{A}_{+}$ has features of an open convex set and therefore contains no extreme points.
We leave this question for future work, as this has no bearing on the validity of our results.


Hence, in order to find extreme points of $\mathcal{D}_+$, we instead look for extreme points of some  subsets of $\mathcal{A}_{+}$.
We consider simple subsets of the set of bounded support WPQS $\mathcal{A}_+^{n}$, whose extreme points are easily identified.
However, the full characterization of $\mathrm{Extr}(\mathcal{A}_+^{n})$ happens to be intractable (as we will see, this is a consequence of the absence of a characterization of non-negative polyanalytic polynomials). Hence, we focus on the smaller subset   $\mathcal A_{\oplus}^n$, i.e., the Fock-bounded WPQS that are phase-invariant. The problem of finding $\mathrm{Extr}(\mathcal{A}_{\oplus}^{n})$ for a given $n$ then becomes tractable by using elementary symmetric polynomials.

\subsection{Bounded support quasi-states (\texorpdfstring{$\mathcal A_+^n$}{})}

Let us focus on a Fock-bounded subset of quasi-states, namely the quasi-states with support in the Fock basis up to (and including) $n\in\mathbb N$, denoted by $\mathcal{A}^{n}$.  Consider a quasi-state $\hat{A}\in\mathcal{A}^{n}$, with Fock matrix elements $A_{k\ell}=\bra{k}\hat{A}\ket{\ell}$.  The Wigner function of $\hat{A}$ is here expressed as:
\begin{align}
    W_{\hat A}(\alpha,\alpha^{\ast})
     & =
    \sum\limits_{k,\ell=0}^{n}
    A_{k\ell}\
    W_{k\ell}(\alpha,\alpha^{\ast})
    \\
     & =
    \mathrm{\exp}(-2\abs{\alpha}^2)
    \underbrace{
        \sum\limits_{k,\ell=0}^{n}
        P_{k\ell}
        \,
        \alpha^{\ast k}\alpha^{\ell}}_{P(\alpha,\alpha^{\ast})}\,, \label{Wigner_is_poly_times_Gauss}
\end{align}
where the scalars $P_{k\ell} \in \mathbb C$ are uniquely determined by $\hat{A}$.  These numbers can be seen as a regrouping of the expansion coefficients of the Wigner function in the Laguerre basis (see Appendix~\ref{apd:vertigo_in_state_space}).
Hence the Wigner function of $\hat{A} \in \mathcal A^n$ in Eq. \eqref{Wigner_is_poly_times_Gauss} is a particular polynomial $P(\alpha,\alpha^{\ast})$ with a Gaussian envelope.  As such, this polynomial fully determines the Wigner-positive or Wigner-negative character of $\hat{A}$.

Hereafter, we note the distinction between a polynomial of a complex variable $\alpha$, which we write as $p(\alpha)=\sum_{k} p_k\alpha^{k}$, and a \textit{polyanalytic polynomial} of a complex variable $\alpha$ \cite{Balk1970-ay, Balk1991-ga}, which we write as $p(\alpha,\alpha^{\ast}) = \sum_{k\ell}P_{k\ell}\alpha^{k}\alpha^{\ast\ell}$.  We also introduce a compact matrix notation for polyanalytic polynomials as follows.  Define the $n$-dimensional complex-valued column vector $\bm{\alpha}\vcentcolon=(1,\alpha,\alpha^{2},\dots,\alpha^{n})^{\intercal}$ for any $\alpha\in\mathbb{C}$.  From the polyanalytic polynomial $P(\alpha,\alpha^{\ast}) = \sum_{l,\ell}P_{k\ell}\alpha^{k}\alpha^{\ast\ell}$, define the matrix $\bm{P}=(P_{k\ell})_{0\leq k,\ell\leq n}$, which we call the monomial matrix of $\hat{A}$.
We have then:
\begin{align}
    P(\alpha,\alpha^{\ast})
      & =
    \sum\limits_{k=0}^{n}
    \;
    \sum\limits_{\ell=0}^{n}
    \
    P_{k\ell}
    \
    \alpha^{\ast k}
    \
    \alpha^{\ell}
    \\[0.5em]
    = &
    \renewcommand{\arraystretch}{0.8}
    \begin{pmatrix}
        1 \\ \vdots\\ \alpha^{\ast n}
    \end{pmatrix}^{\intercal}
    \renewcommand{\arraystretch}{0.8}
    \begin{pmatrix}
        P_{00} & \cdots & P_{0n} \\
        \vdots & \ddots & \vdots \\
        P_{n0} & \cdots & P_{nn}
    \end{pmatrix}
    \renewcommand{\arraystretch}{0.8}
    \begin{pmatrix}
        1 \\ \vdots\\ \alpha^{n}
    \end{pmatrix}
    \\[0.8em]
    = & \
    \bm{\alpha}^{\dagger}\,
    \bm{P}\,
    \bm{\alpha}.
\end{align}
For the monomial matrix $\bm{P}$ to represent a quasi-state, it must define a real and normalized Wigner function. Reality of $W_{\hat{A}}$ corresponds to Hermiticity of $\bm{P}$, and normalization imposes
\begin{align}
    \int\; \bm{\alpha}^{\dagger}\bm{P}\bm{\alpha}\  \mathrm{exp}({-2|\alpha|^2})\ \mathrm{d}^2\alpha = 1,
    \label{eq:normalization_polyanalytic_polynoial}
\end{align}
an underlying assumption we shall make henceforth.
Using the generalized pure loss channel \eqref{eq:def_plc} we can go back and forth between quasi-states and monomial matrices with the following rules:

\begin{align}
    P_{k\ell}
    \ = & \
    \frac{2}{\pi}\
    \sqrt{\frac{2^{k+\ell}}{k!\ell!}}
    \
    \bra{k}
    \mathcal{E}_2
    \big[
        \hat{A}
        \big]
    \ket{\ell},
    \\[0.8em]
    \hat{A}
    \ = & \
    \frac{\pi}{2}\
    \mathcal{E}_{\frac12}
    \left[
        \sum\limits_{k,\ell=0}^{n}
        P_{k\ell}
        \
        \sqrt{\frac{k!\ell!}{2^{k+\ell}}}
        \
        \ket{k}\!\bra{\ell}
        \right].
\end{align}

Identifying the extreme points of the set $\mathcal{A}^{n}_{+}$ then amounts to finding the extreme points of the set of matrices $\bm{P}$ satisfying the inequality $\bm{\alpha}^{\dagger}\bm{P}\bm{\alpha}\geq 0\ \forall\alpha\in\mathbb{C}$.  Note that this condition is weaker than the matrix $\bm P$ being positive semi-definite, as we explore shortly.

\subsubsection{Absolute square of polynomial}

\begin{table}[]
    \centering
    \resizebox{\columnwidth}{!}{%
        \renewcommand{\arraystretch}{1.4}
        \begin{tabular}{|c|c|c|c|}
            \hline
            Set                        & $P(\alpha,\alpha^{\ast})$    & $\bm{P}$                                       & Examples                     \\ \hhline{|=|=|=|=|}
            $\mathcal{A}^{n}_{\sigma}$ & $\abs{p(\alpha)}^2$          & $\bm{q}\bm{q}^{\dagger}$                       & $\hat{\sigma}(\psi,0)$       \\ \hline
            $\mathcal{A}^{n}_{\Sigma}$ & $\abs{p(\alpha,\alpha^*)}^2$ & $\bm{Q}\ast\bm{Q}^{\dagger}$                   & $\hat{\sigma}(\psi,\varphi)$ \\ \hline
            $\mathcal{A}^{n}_{+}$      & $p(\alpha,\alpha^*)\geq 0$   & $\bm{\alpha}^{\dagger}\bm{P}\bm{\alpha}\geq 0$ & $\hat{A}_{\mathrm{Motzkin}}$ \\ \hline
        \end{tabular}%
    }
    \caption{\small
        \textbf{Subsets of }$\mathcal{A}_{+}^{n}$\textbf{.}
        The set $\mathcal{A}_{\sigma}^{n}$ contains quasi-states such that their associated matrix $\bm{P}$ is PSD.
        The set $\mathcal{A}_{\Sigma}^{n}$ contains quasi-states such that their associated polynomial is a finite sum of squared-moduli of polyanalytic polynomials.
        The examples illustrate the inclusion relation $\mathcal{A}^{n}_{\sigma}\subsetneq\mathcal{A}^{n}_{\Sigma}\subsetneq\mathcal{A}^{n}_{+}$.
    }
\end{table}

Clearly, positive semi-definiteness of $\bm{P}$ is a sufficient condition for the non-negativity of $\bm{\alpha}^{\dagger}\bm{P}\bm{\alpha}$.
Extreme points of the set of PSD matrices are Hermitian rank-1 matrices of the form $\bm{P}=\bm{q}\bm{q}^{\dagger}$, where $\bm{q}\in\mathbb{C}^{n}$ is an arbitrary vector.  We then have $\bm{\alpha}^{\dagger}\bm{q}\bm{q}^{\dagger}\bm{\alpha} = \vert\bm{q}^{\dagger}\bm{\alpha}\vert^{2}=\vert\sum_{n}q_{n}^{\ast}\alpha^{n}\vert^2\equiv\vert p(\alpha)\vert^2\geq 0$.
Hence $\bm P$ being PSD corresponds to $P(\alpha,\alpha^*)$ being a convex mixture of absolute squares of polynomials of the form $p(\alpha)=\sum q^{\ast}_n\alpha^{n}$.

Note that any function of the form $\mathrm{exp}(-\abs{\alpha}^2)\abs{p(\alpha)}^2$ defines a Husimi function of some pure state $\ket{\psi}$ \cite{chabaud2020stellar}.
Eq.~\eqref{eq:WtoQ} then implies that such a function is also the Wigner function of the beam-splitter state $\hat{\sigma}(\psi,0)$ up to rescaling.
Defining $\mathcal{A}^{n}_{\sigma}\subset\mathcal{A}^{n}_{+}$ as the set of Fock-bounded quasi-states whose associated polynomial $P(\alpha,\alpha^*)$ is a finite sum of absolute squares of polynomials, we then have:
\begin{align}
    \mathrm{Extr}\left(\mathcal{A}^{n}_{\sigma}\right)
     & =
    \left\lbrace
    \hat{\sigma}(\psi,0)
    \right\rbrace_{\psi\in\mathcal{H}^{n}}
    \\[0.8em]
    \mathcal{A}^{n}_{\sigma}
     & =
    \left\lbrace
    \hat{\sigma}(\hat{\rho},0)
    \right\rbrace_{\hat{\rho}\in\mathcal{D}^{n}}.
\end{align}
As a consequence, all quasi-states with such structure are quantum states, i.e.\ $\mathcal{A}^{n}_{\sigma}\subset\mathcal{D}_{+}$.
Note that the states $\hat{\sigma}(\psi,0)$ are provably extreme within $\mathcal{A}^{n}_{\sigma}$, but it is unknown whether they remain extreme within $\mathcal{D}_{+}$.

\subsubsection{Absolute square of polyanalytic polynomial}
The vector $\bm{\alpha}=(1,\alpha,\alpha^2,...,\alpha^{n})^{\intercal}$ has a particular structure allowing $\bm{\alpha}^{\dagger}\bm{P}\bm{\alpha}$ to be non-negative even for some non-PSD matrices $\bm{P}$.
As a simple example, consider the matrix $\bm{P}=\mathrm{diag}(1,-2,1)$.
Despite not being PSD, it is associated with the polynomial $1-2\abs{\alpha}^2+\abs{\alpha}^{4}=\abs{1-\alpha\alpha^{\ast}}^2\geq 0$.
This particular example illustrates that any polyanalytic polynomial of the form $p(\alpha,\alpha^{\ast})=\abs{q(\alpha, \alpha^{\ast})}^2$ is necessarily non-negative.
Notably, all beam-splitter states $\hat{\sigma}(\psi, \varphi)$ fall into this category, since their Wigner function corresponds to the absolute square of a cross Wigner function \cite{Van-Herstraeten2024-lc}.  For instance, the non-PSD matrix $\bm P = \text{diag}(1, -4, 4)$ corresponds to the beam-splitter state $\hat \sigma(1,1) =  \frac12(\ketbra{0}{0} + \ketbra{2}{2})$.

Defining $\mathcal{A}^{n}_{\Sigma} \subset \mathcal{A}^{n}_{+}$ as the set of Fock-bounded quasi-states whose associated polynomial
$P(\alpha,\alpha^{*})$ can be written as a finite sum of absolute squares  of polyanalytic polynomials, subject to the normalization condition~\eqref{eq:normalization_polyanalytic_polynoial}, we get a larger subset of Fock-bounded WPQS:
clearly, $\mathcal{A}^{n}_{\sigma} \subset \mathcal{A}^{n}_{\Sigma}$.
By definition of extremality, $\mathrm{Extr}\left(\mathcal{A}^{n}_{\Sigma}\right)$   consists of those quasi-states for which $P(\alpha,\alpha^{*})$
reduces to the absolute square of a single (normalized) polyanalytic polynomial, that is, those admitting no nontrivial decomposition into a sum of such terms.

Let us briefly discuss the monomial matrix of an operator associated with the absolute square of a polyanalytic polynomial. Let $\bm{A}$ and $\bm{B}$ be matrices of size $(n+1)\times (n+1)$ and $(m+1)\times (m+1)$, respectively. A straightforward calculation shows that
\begin{align}
    \left(
    \bm{\alpha}^{\dagger}
    \bm{A}
    \bm{\alpha}
    \right)
    \cdot
    \left(
    \bm{\alpha}^{\dagger}
    \bm{B}
    \bm{\alpha}
    \right)
    =
    \bm{\alpha}^{\dagger}
    \left(
    \bm{A}\ast\bm{B}
    \right)
    \bm{\alpha},
    \label{eq:polynomial-product-is-matrix-convolution}
\end{align}
where $\bm{A}\ast\bm{B}$ is a $(n+m+1)\times (n+m+1)$ matrix defined by $(\bm{A}\ast\bm{B})_{k\ell}=\sum_{k',\ell'}A_{k',\ell'}B_{k-k',\ell-\ell'}$, with $0\leq k,l\leq n+m$. Here, $\bm{\alpha}$ is understood to be of the appropriate dimension required by the matrix it is acted upon by.
Equation~\eqref{eq:polynomial-product-is-matrix-convolution} implies that the monomial matrix of any quasi-state in $\mathcal{A}_{\Sigma}^n$ takes the form $\bm{P} = \sum_i \bm{Q}_i \ast \bm{Q}_i^{\dagger}$, where each $\bm{Q}_i$ arises from a polyanalytic component of the polynomial decomposition, and where $\bm{P}$ satisfies the normalization
Equation~\eqref{eq:polynomial-product-is-matrix-convolution} implies that the monomial matrix of any quasi-state in $\mathcal{A}_{\Sigma}^n$ takes the form $\bm{P} = \sum_i \bm{Q}_i \ast \bm{Q}_i^{\dagger}$, where each $\bm{Q}_i$ arises from a polyanalytic component of the polynomial decomposition, and where $\bm{P}$ satisfies the normalization
condition~\eqref{eq:normalization_polyanalytic_polynoial}. Although it goes beyond all beam-splitter states $\hat{\sigma}(\psi, \varphi)$, this set of WPQS generated with the absolute square of a polyanalytic polynomial remains well controlled.
Note that it is unknown whether the states $\hat{\sigma}(\psi,\varphi)$ are extreme within $\mathcal{A}^{n}_{\Sigma}$, and even more so within $\mathcal{D}_{+}$.

\subsubsection{Beyond absolute squares}

The situation would remain simple if every non-negative polynomial was a finite sum of absolute squares of polyanalytic polynomials — but this is not the case.
To illustrate this, consider the Motzkin polynomial \cite{Shisha1967-xv, Reznick1989-ec} defined as $p_{\mathrm{Motzkin}}(x,y)=x^4y^2+x^2y^4-3x^2y^2+1$.
This polynomial has the property to be non-negative while not being a sum of squares.
Setting $x=\alpha+\alpha^{\ast}$ and $y=-i(\alpha-\alpha^{\ast})$, we convert $p_{\mathrm{Motzkin}}$ into a polyanalytic polynomial with monomial matrix:
\begin{align}
    \renewcommand{\arraystretch}{0.8}
    \bm{P}_{\mathrm{Motzkin}}=
    \begin{pmatrix}
        1     & \cdot & \cdot & \cdot & 3     & \cdot \\
        \cdot & \cdot & \cdot & \cdot & \cdot & -4    \\
        \cdot & \cdot & -6    & \cdot & \cdot & \cdot \\
        \cdot & \cdot & \cdot & 8     & \cdot & \cdot \\
        3     & \cdot & \cdot & \cdot & \cdot & \cdot \\
        \cdot & -4    & \cdot & \cdot & \cdot & \cdot \\
    \end{pmatrix}
\end{align}
where the empty entries are zero.
Thus, the monomial matrix $\bm{P}_{\mathrm{Motzkin}}$ yields a Wigner-positive quasi-state that does not belong to $\mathcal{A}_{\Sigma}$, implying the strict inclusion $\mathcal{A}_{\Sigma}\subsetneq\mathcal{A}_{+}$.

Computing the associated quantum operator to $\bm{P}_{\mathrm{Motzkin}}$ yields the strict quasi-state $\hat{A}_{\mathrm{Motzkin}}$, i.e.\ $\hat{A}_{\mathrm{Motzkin}}\not\in\mathcal{D}$.
In fact, we have found no example of Wigner-positive state belonging to $\mathcal{A}_{+}\setminus\mathcal{A}_{\Sigma}$.
If no such state exists, this would mean that $\mathcal{D}_{+}$ is equal to
$\mathcal{A}_{\Sigma}\cap\mathcal{D}$, and substantially simplify the characterization of the WPS set. However, without a proof of this fact, we have followed another path, consisting in adding another restriction on the considered subset of quasi-states.

\subsection{Phase-invariant quasi-states (\texorpdfstring{$\mathcal A_{\oplus}^n$}{})}

As we have seen, the situation is complex when it comes to identifying the extreme points of $\mathcal{A}^{n}_{+}$.
However, when considering its subset of phase-invariant operators, denoted $\mathcal{A}^{n}_{\oplus}$, the problem becomes easier to deal with, as we show hereafter.
In phase space, phase-invariant operators are associated with radial distributions, which are thus described by a univariate polynomial in the real radial variable $r=\abs{\alpha}$ for Fock-bounded operators.
Univariate polynomials can always be factorized, allowing for a simple characterization.

Before diving into the characterization of $\mathrm{Extr}(\mathcal{A}^{n}_{\oplus})$, we should make a preliminary observation:
the extreme points of $\mathcal{A}^{n}_{\oplus}$ are not guaranteed to be extreme points of $\mathcal{A}^{n}_{+}$.
We address this concern in the following Lemma, where we show that extreme phase-invariant quasi-states with a non-zero $n$-particle component remain extreme when phase invariance is relaxed.
\begin{lemma}
    \label{lemma:extreme-fock-diagonal-quasi-states-are-extreme}
    $\mathrm{Extr}(\mathcal{A}^{n}_{\oplus})\setminus\mathcal{A}_{\oplus}^{n-1}\subset\mathrm{Extr}(\mathcal{A}^{n}_{+})$
\end{lemma}

The proof of this result uses the characterization of $\mathrm{Extr}(\mathcal{A}_{\oplus}^n)$ (see Theorem~\ref{th:extrAo+n} below) and is deferred to Appendix~\ref{app:extrAo+n}.

Lemma~\ref{lemma:extreme-fock-diagonal-quasi-states-are-extreme} motivates the interest to study the extreme points of $\mathcal{A}^{n}_{\oplus}$.
Hereafter, we identify all the extreme points of $\mathcal{A}^{n}_{\oplus}$ for arbitrary $n$, using elementary symmetric polynomials $e_{j}$ in $m$ variables given by
\begin{align}
    e_j(\bm{\mu})
    =
    \sum\limits_{1\leq k_1<k_2<\dots< k_j\leq m}
    \mu_{k_1}\mu_{k_2}\cdots\mu_{k_j}.
\end{align}
This result is encapsulated in the following theorem.

\begin{theorem}[Extreme points of $\mathcal{A}^{n}_{\oplus}$]
    \label{th:extrAo+n}
    The following propositions are equivalent:
    \begin{itemize}
        \item[(i)]\
              $\hat{A}\in\mathrm{Extr}(\mathcal{A}^{n}_{\oplus})$.
        \item[(ii)]\ $W_{\hat{A}}\big(\alpha\big)=P\big(\abs{\alpha}^2\big)\;\mathrm{exp}(-2\abs{\alpha}^2)$ where $P$ is a polynomial of the form $P(t)=c\  t^{k}\prod_{i=1}^{\ell}(t-\lambda_i)^{2}$ such that $\lambda_i>0\ \forall i$, $k+2\ell\leq n$ and $c$ is a positive normalization constant.
        \item[(iii)]\
              $\hat{A}=(-1)^{m}\sum_{j=0}^{m}(-1)^{j}j!\, e_{m-j}(\bm{\mu})\,\hat{\sigma}(j,0)$, where $e_{j}$ is the ${j}^{\textit{th}}$ elementary symmetric polynomial and the vector $\bm{\mu}$ has $m\leq n$ non-negative components, such that strictly positive components come by pairs, i.e.\ $\bm{\mu}=\big(\underbrace{0,...,0}_{k},\underbrace{\lambda_1,\lambda_1,...,\lambda_\ell,\lambda_\ell}_{2\ell}\big)$.
    \end{itemize}
\end{theorem}

We provide the proof in Appendix~\ref{app:extrAo+n}.
Theorem~\ref{th:extrAo+n} immediately yields the inclusion chain
$\mathrm{Extr}(\mathcal{A}^{n}_{\oplus}) \subset \mathrm{Extr}(\mathcal{A}^{n+1}_{\oplus})$.
This nesting of extreme sets does not carry over to $\mathcal{A}_{+}$: in general, $\mathrm{Extr}(\mathcal{A}_{+}^{n}) \not\subset \mathrm{Extr}(\mathcal{A}^{n+1}_{+})$, as discussed in Appendix~\ref{app:extrAo+n}.
Finally, we observe that every extreme point of $\mathcal{A}^{n}_{\oplus}$ arises from the absolute square of a
polyanalytic polynomial, which in turn implies the inclusion $\mathcal{A}^{n}_{\oplus} \subset \mathcal{A}^{n}_{\Sigma}$.

Having characterized a large class of extreme WPQS, we introduce in the following section a map which sends extreme WPQS to extreme WPS.

\section{The Vertigo map: from quasi-states to states}
\label{sec:vertigo}

This section is devoted to a central contribution of this paper: the definition of a new quantum map -- the \textit{Vertigo map} -- whose action is carefully designed to rescale the polynomial part of the Wigner function of Fock-bounded operators.
As we are going to show, it plays a particular role with respect to the set of Wigner-positive states and quasi-states.

\subsection{Definition and properties}
We define the Vertigo map $\mathcal{V}_{t}$ as the probabilistic map that acts on Wigner functions as follows:
\begin{align}
    W_{\mathcal{V}_t[\hat{\rho}]}(\alpha)
    =
    W_{\hat{\rho}}(\sqrt{t}\alpha)\
    \exp(2(t-1)\abs{\alpha}^2),
    \label{eq:def_vertigo_map_wigner}
\end{align}
where $t\geq 0$ is a tunable rescaling parameter.
The map $\mathcal{V}$ is defined in such a way that it rescales the polynomial part of a Fock-bounded operator quasi-state $\hat{A}\in\mathcal{A}^n$.
Such an operator $\hat{A}$ has a Wigner transform of the form $W_{\hat{A}}(\alpha)=P_{\hat{A}}(\alpha,\alpha^{\ast})\;\mathrm{exp}(-2\abs{\alpha}^2)$, where $P_{\hat{A}}$ is a polyanalytic polynomial of degree $n$ in $\alpha$ and $\alpha^{\ast}$.
More precisely, when the operator $\hat{A}$ evolves under the Vertigo map $\mathcal{V}_{t}$,
the action on its associated Wigner function $W(\alpha)$,
polyanalytic polynomial $P_{\hat{A}}(\alpha,\alpha^{\ast})$, and monomial matrix $\bm{P}$ is as follows:

\begin{align*}
    \hat{A}
    \quad & \longrightarrow\quad
    \mathcal{V}_{t}\big[\hat{A}\big]
    \\
    W(\alpha)
    \quad & \longrightarrow\quad
    W(\sqrt{t}\alpha)\exp(2(t-1)\abs{\alpha}^2)
    \\
    P(\alpha,\alpha^{\ast})
    \quad & \longrightarrow\quad
    P(\sqrt{t}\alpha,\sqrt{t}\alpha^{\ast})
    \\
    (\bm{P})_{k\ell}
    \quad & \longrightarrow\quad
    \sqrt{t}^{k+\ell}
    (\bm{P})_{k\ell}.
\end{align*}

The choice to name this map ‘‘Vertigo'' stems from its dual, opposing actions.
As indicated in Eq.~\eqref{eq:def_vertigo_map_wigner}, when
$t>1$, the Wigner function is compressed toward the origin of phase space.
Simultaneously, it is multiplied by a function that amplifies the contribution of regions further away from the origin.
The resulting action of these two effects on the Wigner function is to rescale its polynomial part, while preserving exactly its exponential decay.
This mechanism resembles the cinematic ‘‘Vertigo effect'' (also known as ‘‘Dolly zoom'') famously used in Hitchcock's film \cite{hitchcock1958, Federico2012-rr}.
The Vertigo effect consists in zooming in (resp.\ out) the frame while moving the camera away (resp.\ closer).
Doing so, the subject stays roughly in the same size in the frame, while the background shrinks or expands.

As we show in Appendix~\ref{apd:vertigo_in_state_space}, the action of the Vertigo map, described in the phase-space picture by Eq.~\eqref{eq:def_vertigo_map_wigner}, can be formulated in state space as:
\begin{align}
    \mathcal{V}_t\big[\hat{\rho}\big]
    =
    \sum\limits_{k=0}^{\infty}
    \
    \frac{(t-1)^k}{k!}
    \sqrt{t}^{\,\hat{n}}
    \hat{a}^{k}
    \;
    \hat{\rho}
    \;
    \hat{a}^{\dagger k}
    \sqrt{t}^{\,\hat{n}},
    \label{eq:def_vertigo_map_channel}
\end{align}
which provide an equivalent definition of the map $\mathcal{V}_{t}$.

The Vertigo map can be decomposed as the composition of a Pure Loss Channel (PLC) and Noiseless Linear Amplifier (NLA) (see~\cref{fig:vertigo_physical_implementation}).
The PLC is the channel $\mathcal{E}_{\eta}[\hat{\rho}]=\mathrm{Tr}_2\big[\hat{U}_{\eta}(\hat{\rho}\otimes\ket{0}\!\bra{0})\hat{U}^{\dagger}_{\eta}\big]$ where $\eta$ is the transmittance, and the NLA is the map $\mathcal{N}_{g}[\hat{\rho}]=\sqrt{g}^{\,\hat{n}}\,\hat{\rho}\,\sqrt{g}^{\,\hat{n}}$ where $g$ is the gain.
We have the relation:
\begin{align}
    \mathcal{V}_{t}
    = & \;
    \mathcal{E}_{\eta}
    \circ
    \mathcal{N}_{g},
\end{align}
where the transmittance of the PLC is $\eta=t/(2t-1)$ and the gain of the NLA is $g=2t-1$.
Note that this decomposition is not unique (see Appendix~\ref{apd:vertigo_in_state_space}).
Note also that the Vertigo map acts on the Wigner function the same way as the NLA acts on the Husimi function \cite{Ralph2009-ue, Gagatsos2014-xa, Zhao2017-fj}.
This translates into the fact that $\mathcal{V}_{t}[\hat{\sigma}(\hat{\rho},0)]=\hat{\sigma}(\mathcal{N}_{t}[\hat{\rho}],0)$.

A few observations should be made about $\mathcal{V}_{t}$.
First, observe that it is not a trace-preserving map, since the following can be computed: $\mathrm{Tr}[\mathcal{V}_t[\hat{\rho}]]=\mathrm{Tr}[(2t-1)^{\hat{n}}\hat{\rho}]$.
In fact, for some inputs the trace even diverges, e.g.~for thermal states of high temperature.
From Eq.~\eqref{eq:def_vertigo_map_channel}, we see that the Vertigo map never creates higher components in the Fock basis, and preserves the Fock support of the input.
For an input $\hat{A}\in\mathcal{A}^{n}$, the trace of the output always converges, and we focus on such inputs in the following.
Because $\mathcal{V}_{t}$ is not trace-preserving, it does not map quasi-states to quasi-states.
In the following we introduce a normalized, trace-preserving version of $\mathcal{V}_{t}$, but before doing so let us investigate when we can divide by the trace.
Using Eq.~\eqref{eq:def_vertigo_map_wigner} and $\mathrm{Tr}[\hat{A}]=\pi\int W_{\hat{A}}(\alpha)\mathrm{d}^2\alpha$, it is apparent that applying the Vertigo map to a Wigner-positive quasi-state yields an operator whose trace remains strictly positive.
More precisely, for all $\hat{A}\in\mathcal{A}^{n}_{+}$ we have that $\mathrm{Tr}\big[\mathcal{V}_{t}[\hat{A}]\big]>0$ for all $t>0$.
We define the notation $\mathcal{V}^{\mathrm{norm}}_t[\hat{\rho}]\vcentcolon=\mathcal{V}_t[\hat{\rho}]/\mathrm{Tr}[\mathcal{V}_t[\hat{\rho}]]$.
We can then write:
\begin{align}
    \mathcal{V}^{\mathrm{norm}}_{t}:
    \mathcal{A}^{n}_{+}\rightarrow\mathcal{A}^{n}_{+},
\end{align}
which holds for all $t>0$.
In the limiting case $t=0$, the trace of $\mathcal{V}_{t}[\hat{A}]$ vanishes if $W_{\hat{A}}(0)=0$, in which case the map is ill-defined.
We also note that Wigner-negative quasi-states are mapped to Wigner-negative quasi-states; that is,
$\mathcal{V}_{t}^{\mathrm{norm}}$ sends $\mathcal{A}^{n}\setminus\mathcal{A}^{n}_{+}$ into itself for every $t\in(0,\infty)$,
as long as the trace does not vanish (otherwise the map is ill defined).


Second, as it appears from Eq.~\eqref{eq:def_vertigo_map_channel}$, \mathcal{V}_t$ preserves positive semi-definiteness when $t\geq 1$, but not necessarily when $0<t<1$. In particular, for $t\geq 1$ we have
\begin{align}
    \mathcal{V}^{\mathrm{norm}}_{t\geq 1}:
    \mathcal{D}^{n}\rightarrow\mathcal{D}^{n}.
\end{align}

\subsection{Vertigo trajectories}

Observe that $(\lbrace\mathcal{V}_{t}\rbrace_{t\in\mathbb{R}_{>0}},\circ)$ has the same group structure as $(\mathbb{R}_{>0},\times)$.
In particular, it is closed under composition:
\begin{align}
    \mathcal{V}^{\mathrm{norm}}_{t_2}\circ\mathcal{V}^{\mathrm{norm}}_{t_1}=\mathcal{V}^{\mathrm{norm}}_{t_1 t_2}.
\end{align}
This allows us to define the \textit{Vertigo trajectory} of a quantum operator $\hat{A}$ as the set of operators $\big\lbrace\mathcal{V}_{t}^{\mathrm{norm}}[\hat{A}]\big\rbrace_{t>0}$.
We refer to the subset corresponding to $t \geq 1$ as the \textit{forward} trajectory, and to the subset corresponding to $0 < t \leq 1$ as the \textit{backward} trajectory.
Remarkably, the Vertigo trajectory of an extreme WPQS contains only extreme WPQS, as per the forthcoming lemma.

\begin{lemma}[Extreme Vertigo trajectories]
    \label{lemma:vertigo_extreme_trajectories}
    If $\hat{A}\in\mathrm{Extr}(\mathcal{A}^{n}_{+})$, then $\mathcal{V}^{\mathrm{norm}}_t[\hat{A}]\in\mathrm{Extr}(\mathcal{A}^{n}_{+})$ for all $t>0$.
\end{lemma}

This lemma is a straightforward consequence of the fact that extreme non-negative polynomials remain extreme under rescaling.
Interestingly, some quasi-states are invariant under the Vertigo map $\mathcal{V}^{\mathrm{norm}}_{t}$.
As we show in Appendix~\ref{app:vertigo_traj} (see Lemma~\ref{lem:eig-op}),
such quasi-states are identified with quasi-states associated with a homogeneous polynomial $P(\alpha,\alpha^{\ast})$, i.e.,\ such that $P(\alpha,\alpha^{\ast})=\sum_{k=0}^{N}c_{k}\alpha^{k}\alpha^{\ast N-k}$.
While some of these operators are strict quasi-states, some are states.
Remarkably, all the quantum states with homogeneous polynomial have a Wigner function of the form (see Lemma~\ref{lem:eig-dens-op}):
\begin{align}
    W_{\hat{\sigma}(n,0)}
    (\alpha)
    =
    \frac{2}{\pi}\
    \frac{1}{n!}
    \left(2\abs{\alpha}^2\right)^{n}\
    \exp(-2\abs{\alpha}^2).
    \label{eq:wigner_function_binomial_state}
\end{align}
The above Wigner function is in fact that of a beam-splitter state $\hat{\sigma}(n,0)$, which we call \textit{binomial} beam-splitter state as it is a binomial mixture of Fock states:
\begin{align}
    \hat{\sigma}(n,0)
    =
    \frac{1}{n!}
    \sum\limits_{k=0}^{n}
    \binom{n}{k}
    \ket{k}\!\bra{k}.
\end{align}
Note that these are completely dephased versions of pure binomial states introduced in the context of bosonic quantum error correction \cite{michael2016new}.
Observe that $\hat{\sigma}(n,0)$ has rank $n$ and is a full-rank positive-definite operator in $\mathcal{D}^{n}$, so that $\hat{\sigma}(n,0)\in\mathrm{int}(\mathcal{D}^{n})$.
Since any quasi-state $\hat{A}\in\mathcal{A}^{n}$ such that $\bra{n}\hat{A}\ket{n}>0$ converges to $\hat{\sigma}(n,0)$ (corollary of Lemma~\ref{lem:fixed_point}), this means that there exists a finite $t_0\in\mathbb{R}_{+}$ such that $\mathcal{V}^{\mathrm{norm}}_{t_0}[\hat{A}]\in\mathcal{D}^{n}$.
Similarly, for all $\hat{\rho}\in\mathcal{D}^{n}\setminus\mathcal{D}^{n-1}$:
\begin{align}
    \lim\limits_{t\rightarrow\infty}
    \
    \mathcal{V}^{\mathrm{norm}}_{t}
    \left[
        \hat{\rho}
        \right]
    =
    \hat{\sigma}(n,0).
\end{align}

We can combine the previous observations into the following lemma:

\begin{lemma}[$\mathcal{V}$ maps quasi-states to states]
    \label{lemma:vertigo_quasi-states_to_states}
    Let $\hat{A}$ be a quasi-state of $\mathcal{A}^{n}$ such that $\bra{n}\hat{A}\ket{n}\neq 0$.
    Then, there exists $t_0>0$ such that $\mathcal{V}^{\mathrm{norm}}_t[\hat{A}]\in\mathcal{D}^{n}$, $\forall t\geq t_0$.
\end{lemma}

At this stage, we have established that the Vertigo map possesses a family of fixed points that attract both states and quasi-states as $t \to \infty$. We also showed that these fixed points lie strictly inside the set of density operators, which forces Vertigo trajectories starting in $\mathcal{A}$ to eventually enter $\mathcal{D}$ (Lemma~\ref{lemma:vertigo_quasi-states_to_states}). In addition, the Vertigo map preserves extremality with respect to Wigner-positivity, since extreme WPQS remain extreme along their evolution (Lemma~\ref{lemma:vertigo_extreme_trajectories}). Combining these two observations leads directly to the main result of this section, stated in the following theorem.

\begin{theorem}[$\mathcal{V}$ maps extreme WPQS to extreme WPS]\label{th:vertigoWPQStoWPS}
    Let $\hat{A}$ be an extreme quasi-state of $\mathcal{A}^{n}_{+}$ such that $\bra{n}\hat{A}\ket{n}>0$.
    Then, there exists $t_0>0$ such that $\mathcal{V}^{\mathrm{norm}}_t[\hat{A}]\in\mathrm{Extr}(\mathcal{D}^{n}_{+})\subset\mathrm{Extr}(\mathcal{D}_{+})$, $\forall t\geq t_0$.
    \label{theorem:vertigo-extreme-quasi-states-to-extreme-states}
\end{theorem}

Theorem~\ref{th:vertigoWPQStoWPS} is a straightforward consequence of Lemma~\ref{lemma:vertigo_extreme_trajectories}, Lemma~\ref{lemma:vertigo_quasi-states_to_states} and Eqs.~(\ref{eq:extreme_quasi-states_that_are_states_are_extreme_bounded},~\ref{eq:extreme_quasi-states_that_are_states_are_extreme}).

The strength of Theorem~\ref{th:vertigoWPQStoWPS} is that it provides a direct bridge between the characterization of extreme WPQS and that of extreme WPS. As discussed in Section~\ref{sec:extreme_wigpos_quasistates}, identifying extreme WPQS is a significantly more tractable problem than characterizing extreme WPS directly. In the next section, we exploit Theorem~\ref{th:vertigoWPQStoWPS} together with Theorem~\ref{th:extrAo+n} to obtain a constructive procedure that generates a large family of extreme Wigner-positive states.


\section{Application: generating extreme Wigner-positive states}
\label{sec:generating}

\begin{figure}
    \centering
    \begin{mybox}{Algorithm for generating extreme WPS}\label{box:algorithm}
        \justifying{
            \noindent\textbf{Extreme starting points}
            \begin{enumerate}
                \item[${\bm1.}$] Start from a quasi-state in $\mathrm{Extr}(\mathcal{A}^{n}_{\oplus})$.
                      This set is exactly characterized by Theorem~\ref{th:extrAo+n}.
                \item[${\color{orange}\bm2.}$] Apply $\mathcal{V}^{\mathrm{norm}}_{t}$ beyond the critical~$t_0$.
                      The resulting \textit{state} belongs to $\mathrm{Extr}(\mathcal{D}_{\oplus}^n$), as per Theorem~\ref{theorem:vertigo-extreme-quasi-states-to-extreme-states}.
            \end{enumerate}
            \textbf{Extreme orbits}
            \begin{itemize}
                \item[${\color{ForestGreen}\bm3.}$] Apply $\bar{D}^{\mathrm{norm}}_{\alpha}$ from Eq.~(\ref{eq:fock_bounded_dispalcement_operator}) for any $\alpha\in\mathbb{C}$.
                      This adds coherence to the state while keeping it in $\mathrm{Extr}(\mathcal{D}^{n}_{+})$.
                \item[${\color{Goldenrod}\bm4.}$] Apply any Gaussian unitary $\hat{G}$.
                      This brings the state outside $\mathcal{D}^{n}_{+}$ but keeps it within $\mathrm{Extr}(\mathcal{D}_{+})$.
            \end{itemize}
        }
    \end{mybox}
    \label{fig:placeholder}
\end{figure}

\begin{figure}[t]
    \centering
    \includegraphics[width=\linewidth]{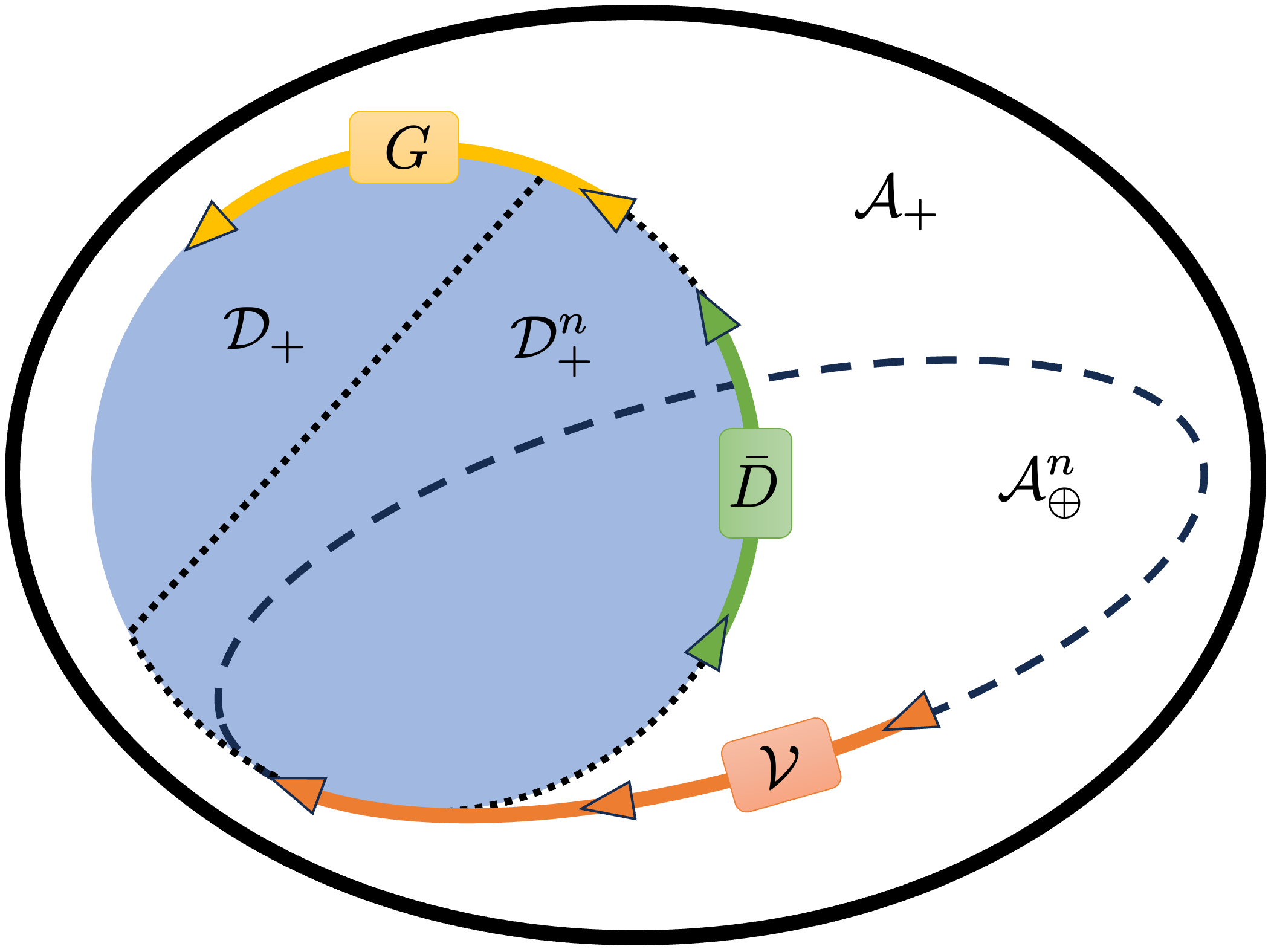}
    \caption{\small \textbf{Pictorial depiction of our method for generating extreme WPS.} We refer to Table~\ref{table:notations_convex_sets} for notations.
        The set of WPS $\mathcal{D}_{+}$ (blue-shaded disc) is contained in the set of WPQS $\mathcal{A}_{+}$ (outer solid-line black ellipse).
        The subset of Fock-bounded phase-invariant WPQS $\mathcal{A}^{n}_{\oplus}$ (inner dashed ellipse) intersects with the set of Fock-bounded WPS $\mathcal{D}^{n}_{+}$ (dotted curve).
        The Vertigo map $\mathcal{V}$ (orange arrows) maps extreme points of $\mathcal{A}^{n}_{\oplus}$ to extreme points of $\mathcal{D}^{n}_\oplus$.
        The Fock-bounded displacement-like operator $\bar{D}$ (green arrows) sends extreme states of $\mathcal{D}^{n}_{\oplus}=\mathcal{A}^{n}_{\oplus}\cap\mathcal{D}^{n}_{+}$ to extreme states of $\mathcal{D}^{n}_{+}$.
        Finally, Gaussian unitary operations $G$ (yellow arrows) send extreme states of $\mathcal{D}^{n}_{+}$ to extreme states of $\mathcal{D}_{+}$.
    }
    \label{fig:extreme_state_generation}
\end{figure}

In this section, we bring together the pieces we have constructed in the previous sections and develop a general method to generate extreme Wigner-positive states.
Our construction relies on our ability to identify classes of extreme WPS and extremality-preserving maps:

$(i)$ Using the properties of the Vertigo map, we produce a large class of extreme states in $\mathrm{Extr}(\mathcal{A}^{n}_{\oplus})\cap\mathcal{D}$.
We detail this step in subsection~\ref{subsec:starting_points}.

$(ii)$ Then, we find quantum maps with the property to preserve the extremality of Wigner-positive states.
These maps trace orbits of extreme Wigner-positive states, while relaxing phase-invariance and Fock-boundedness. This yields extreme Wigner-positive states in $\mathcal{D}_+$, going beyond $\mathcal{D}^{n}_{\oplus}$.
We detail this step in subsection~\ref{subsec:extreme_orbits}.

The algorithm implementing this procedure is summarized above, and illustrated in~\cref{fig:extreme_state_generation,fig:proof_structure}.

\begin{figure*}
    \centering
    \includegraphics[width=0.7\linewidth]{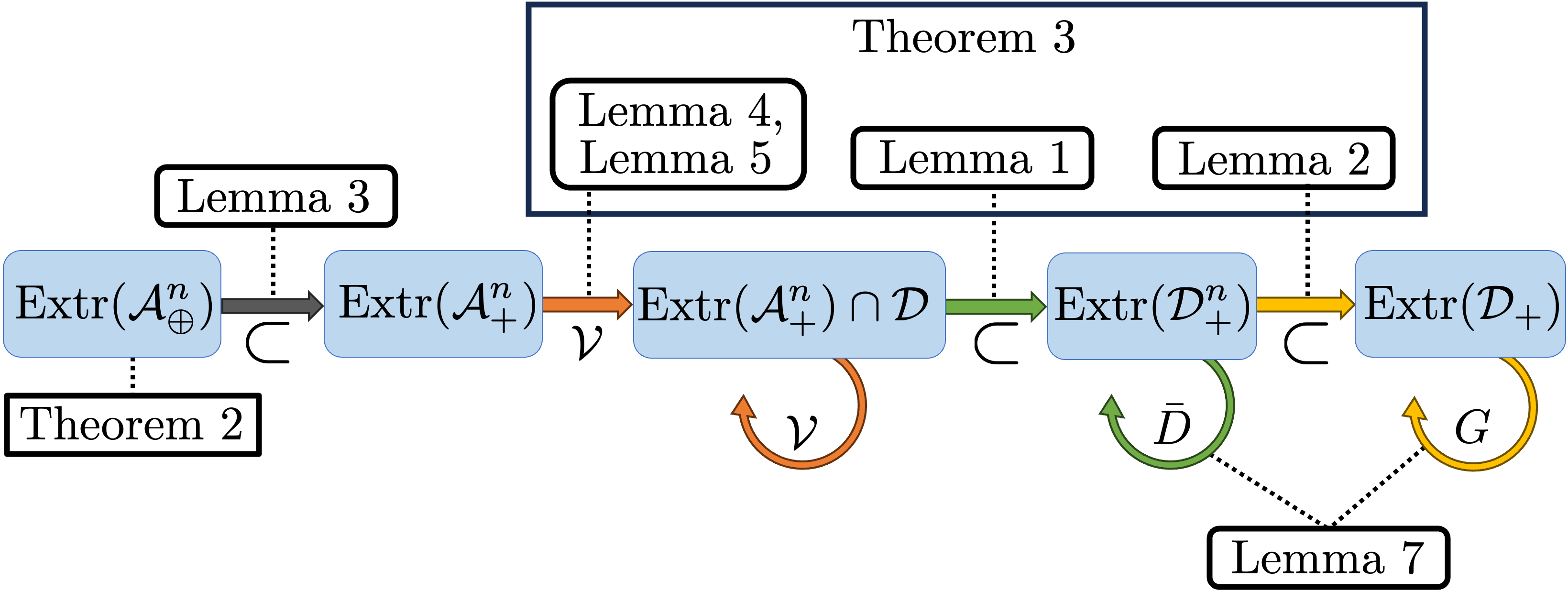}
    \caption{\small \textbf{Chain of results for constructing extreme WPS.}
    Any quasi-state $ \hat{A} \in \mathrm{Extr}(\mathcal{A}^{n}_{\oplus}) $, as characterized by Theorem~\ref{th:extrAo+n}, is also an extreme point of $ \mathcal{A}^{n}_{+} $ (Lemma~\ref{lemma:extreme-fock-diagonal-quasi-states-are-extreme}).
    Under Vertigo evolution, $\hat{A}$ enters the set $\mathcal{D}$ at some finite $t_0$ (Lemma~\ref{lemma:vertigo_quasi-states_to_states}), while remaining extreme in $\mathcal{A}^{n}_{+}$ (Lemma~\ref{lemma:vertigo_extreme_trajectories}).
    Then, for $t\geq t_0$, Lemma~\ref{lemma:extreme_superset} ensures that $ \mathcal{V}^{\mathrm{norm}}_{t}[\hat{A}]$ belongs to $\mathrm{Extr}(\mathcal{D}^{n}_{+})$, and Lemma~\ref{lemma:extremality_preserving_map_v2} guarantees that extremality is preserved under any $ \bar{D} $-orbit.
    Finally, Lemma~\ref{lemma:extreme_subset_projector} shows that the resulting state is also extreme in $ \mathcal{D}_{+} $, and remains so under any Gaussian orbit (Lemma~\ref{lemma:extremality_preserving_map_v2}).
    }
    \label{fig:proof_structure}
\end{figure*}

\subsection{Extreme starting points}
\label{subsec:starting_points}

The results of Section~\ref{sec:extreme_wigpos_quasistates} together with Theorem~\ref{theorem:vertigo-extreme-quasi-states-to-extreme-states} allow us to generate a large class of extreme states.
From Lemma~\ref{lemma:extreme-fock-diagonal-quasi-states-are-extreme}, we know that the extreme points of $\mathcal{A}_{\oplus}^{n}$, as identified from Theorem~\ref{th:extrAo+n}, are also extreme points of $\mathcal{A}^{n}_{+}$.
Then, from Theorem~\ref{theorem:vertigo-extreme-quasi-states-to-extreme-states}, we know that there exists a finite $t_{0}>0$ such that the Vertigo map brings them inside $\mathcal{D}$.
In mathematical terms, we write:
\begin{align}
    \forall\hat{A}\in\mathrm{Extr}(\mathcal{A}^{n}_{\oplus}),
    \quad
    \exists t_{0}>0:
    \nonumber
    \\[0.6em]
    \mathcal{V}^{\mathrm{norm}}_{t}
    \big[
        \hat{A}
        \big]
    \in\mathrm{Extr}(\mathcal{D}^{n}_{\oplus})
    \quad
    \forall t\geq t_0.
\end{align}

The construction presented here requires the use of the Vertigo map in order to ensure that the resulting operator is a valid quantum state.
However, in some cases we can ensure that the starting point is already a quantum state.
Consider the set of beam-splitter states of the form $\hat{\sigma}(m,n)$.
Their Wigner function is:
\begin{align}
    W_{\hat{\sigma}(m,n)}
    (\alpha)
    = &
    \frac{2}{\pi}
    \frac{m!}{n!}
    \left(
    2\abs{\alpha}^2
    \right)^{n-m}
    \exp(-2\abs{\alpha}^2)
    \nonumber
    \\&\times
    \left[
        L^{(n-m)}_{m}
        \left(2\abs{\alpha}^2\right)
        \right]^2,
    \label{eq:wigner_function_bs_state}
\end{align}
as can be computed from the interference formula~\cite{Van-Herstraeten2024-lc} (assuming $n\geq m$).
We see that the Wigner functions described by Eq.~\eqref{eq:wigner_function_bs_state} are associated to polynomials of degree $n+m$ in $\abs{\alpha}^2$, having all their zeros on the non-negative real axis.
Therefore, they meet the condition of Theorem~\ref{th:extrAo+n}, and we conclude that $\hat{\sigma}(m,n)\in\mathrm{Extr}(\mathcal{A}^{m+n}_{\oplus})$. This yields the following theorem:

\begin{theorem}
    The beam-splitter states $\hat{\sigma}(m,n)$ are extreme WPS $\forall m,n\in\mathbb{N}$:
    \begin{align}
        \big\lbrace
        \hat{\sigma}(m,n)
        \big\rbrace_{m,n\in\mathbb{N}}
         & \subset\mathrm{Extr}(\mathcal{D}_{+}).
    \end{align}
    \label{theorem:beam-splitter-states-are-extreme}
\end{theorem}

Interestingly, Eq.~\eqref{eq:wigner_function_bs_state} also provides insight into the trajectories of beam-splitter states under the Vertigo map.
In particular, we observe that the Wigner function of $\hat{\sigma}(m,n)$ exhibits a total of $m+n$ zeros (counting multiplicities), including a zero of multiplicity $\abs{n-m}$ located at the origin.
This structure implies that the Vertigo trajectory $\mathcal{V}^{\mathrm{norm}}_t[\hat{\sigma}(m,n)]$ converges to $\hat{\sigma}(m+n,0)$ in the limit $t \to \infty$, and to $\hat{\sigma}(\abs{m-n},0)$ in the limit $t \to 0$.

Let us mention a last remarkable property of beam-splitter states.
As highlighted by Theorem~\ref{theorem:beam-splitter-states-are-extreme}, the states $\hat{\sigma}(m,n)$ are always extreme WPS.
As such, they extremize linear functionals over the set $\mathcal{D}_{+}^{m+n}$, as illustrated by the following lemma.

\begin{lemma}
    \label{lemma:bs_states_fidelity_fock_states}
    Beam-splitter states maximize the fidelity with Fock states among Wigner-positive states:
    \begin{align}
        \hat{\sigma}\big(\lceil n/2\rceil,\lfloor n/2\rfloor\big)
         & =
        \arg\max_{\hat{\rho} \in \mathcal{D}^{n}_{\oplus}}
        \bra{n} \hat{\rho} \ket{n}
    \end{align}
\end{lemma}

\noindent Note that the optimizations are over phase-invariant states without loss of generality, since Fock states are phase-invariant. This lemma is a consequence of a result from Ref.~\cite{Chabaud2021witnessing} and is proven in Appendix~\ref{app:beam-splitter-states}.

\subsection{Extreme orbits}

\label{subsec:extreme_orbits}

We have now presented two techniques to obtain extreme Wigner-positive states, via the Vertigo map or via beam-splitter states.
Here, we go one step further and present classes of channels that preserve the extremality of Wigner-positive states, and thus produce orbits of extreme Wigner-positive states.
The central ingredient we use to derive these maps is the following lemma.

\begin{lemma}[Extremality-preserving map]
    \label{lemma:extremality_preserving_map_v2}
    Let $\mathcal{M}$ be a linear, invertible (not necessarily trace-preserving) map.
    Assume:
    \begin{enumerate}
        \item $\mathcal{M},\mathcal{M}^{-1}$ preserve positive semi-definiteness,
        \item $\mathcal{M},\mathcal{M}^{-1}$ preserve Wigner positivity.
    \end{enumerate}
    Define $\mathcal{M}^{\mathrm{norm}}[\hat{\rho}]=\mathcal{M}[\hat{\rho}]/\mathrm{Tr}\big[\mathcal{M}[\hat{\rho}]\big]$.
    Then, $\mathcal{M}^{\mathrm{norm}}:\mathrm{Extr}(\mathcal{D}_{+})\rightarrow\mathrm{Extr}(\mathcal{D}_{+})$.
\end{lemma}

As a simple illustration of this lemma, consider Gaussian unitaries.
Channels of the form $\mathcal{M}[\hat{\rho}]=\hat{U}_{\mathrm{G}}\hat{\rho}\hat{U}^{\dagger}_{\mathrm{G}}$, where $\hat{U}_{\mathrm{G}}$ is a Gaussian unitary, satisfy the conditions of
Lemma~\ref{lemma:extremality_preserving_map_v2}.
Thus, if $\hat{\rho}\in\mathrm{Extr}(\mathcal{D}_{+})$, then $\hat{U}_{\mathrm{G}}\hat{\rho}\hat{U}^{\dagger}_{\mathrm{G}}\in\mathrm{Extr}(\mathcal{D}_{+})$.
This simple result is a consequence of the fact that affine symplectic transformations (displacement, rotation, squeezing) do not change the nature of the quantum state with respect to Wigner-positivity.

Following the same idea that lead to the construction of the Vertigo map, we introduce another type of transformation, namely a displacement of the polynomial part of the Wigner function.
We define the Fock-bounded displacement map $\bar{D}_{\beta}$ from its action on Wigner function as follows:
\begin{align}
    W_{\bar{D}_\beta[\hat{\rho}]}(\alpha)
    = & \
    W_{\hat{\rho}}(\alpha-\beta)\nonumber
    \\
      & \times\exp\big(
    2\abs{\alpha-\beta}^2-2\abs{\alpha}^2
    \big).
    \label{eq:def_fock-bounded_disp_wigner}
\end{align}
Remarkably, the action of $\bar{D}_{\beta}$ in operator space is simply described as:
\begin{align}
    \bar{D}_{\beta}\big[\hat{\rho}\big]
    =
    \exp\big(-\beta^{\ast}\hat{a}\big)\;
    \hat{\rho}\;
    \exp\big(-\beta\hat{a}^{\dagger}\big).
    \label{eq:fock_bounded_dispalcement_operator}
\end{align}
as we show in Appendix~\ref{apd:fock-bounded_disp}.
It is clear from Eq.~\eqref{eq:def_fock-bounded_disp_wigner} that $\bar{D}_{\beta}$ preserves Wigner positivity, and from Eq.~\eqref{eq:fock_bounded_dispalcement_operator} that it preserves positive semi-definiteness and the Fock support.
In addition, it is easily seen that $\bar{D}_{\beta}^{-1}=\bar{D}_{-\beta}$.
These observations imply that $\bar{D}_{\beta}$ satisfies the conditions of Lemma~\ref{lemma:extremality_preserving_map_v2}, so that, after defining $\bar{D}^{\mathrm{norm}}_{\beta}[\hat{\rho}]=\bar{D}_{\beta}[\hat{\rho}]/\mathrm{Tr}\big[\bar{D}_{\beta}[\hat{\rho}]\big]$, we have $\bar{D}^{\mathrm{norm}}_{\beta}:\mathrm{Extr}(\mathcal{D}^{n}_{+})\rightarrow\mathrm{Extr}(\mathcal{D}^{n}_{+})$.


\subsection{Low-dimensional extreme points}

As we show in this section, the tools we have developed so far are, in some low-dimensional cases, sufficient to generate the whole set of extreme WPS.

\begin{figure}
    \centering
    \makebox[\linewidth]{\includegraphics[width=1.1\linewidth]{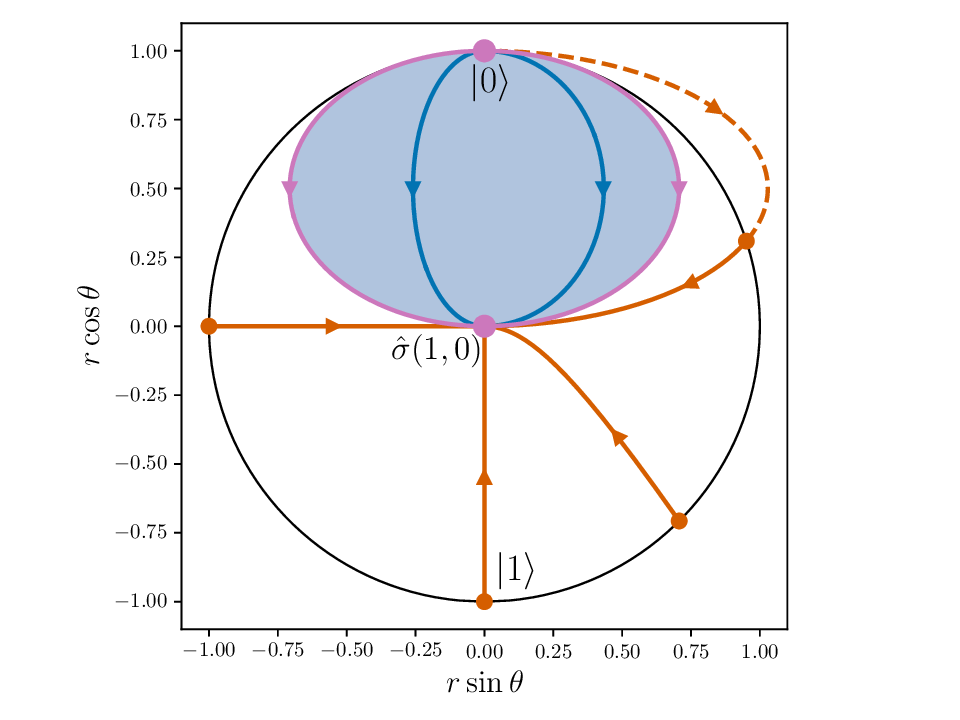}}
    \caption{\small
        Slice of the Bloch sphere representing $\mathcal{D}^{1}$, where the blue area denotes $\mathcal{D}^{1}_{+}$. Each point corresponds to a state $\hat{\rho}=r\,\ket{\psi(\theta)}\!\bra{\psi(\theta)}+(1-r)(\ket{0}\!\bra{0}+\ket{1}\!\bra{1})/2$, with $\ket{\psi(\theta)}=\cos(\theta/2)\ket{0}+\sin(\theta/2)\ket{1}$. The North and South poles correspond to $\ket{0}$ and $\ket{1}$, and the centre is the binomial state $\hat{\sigma}(1,0)$. Vertigo trajectories are shown as curves with arrows indicating their direction. Blue trajectories correspond to Wigner-positive states, orange trajectories to Wigner-negative \mbox{(quasi-)}states, and purple trajectories to the beam-splitter states forming the boundary of $\mathcal{D}^{1}_{+}$. Solid lines denote states; dashed lines denote quasi-states.
    }
    \label{fig:bloch_sphere}
\end{figure}

\begin{figure}
    \centering
    \makebox[\linewidth]{\includegraphics[width=1.1\linewidth]{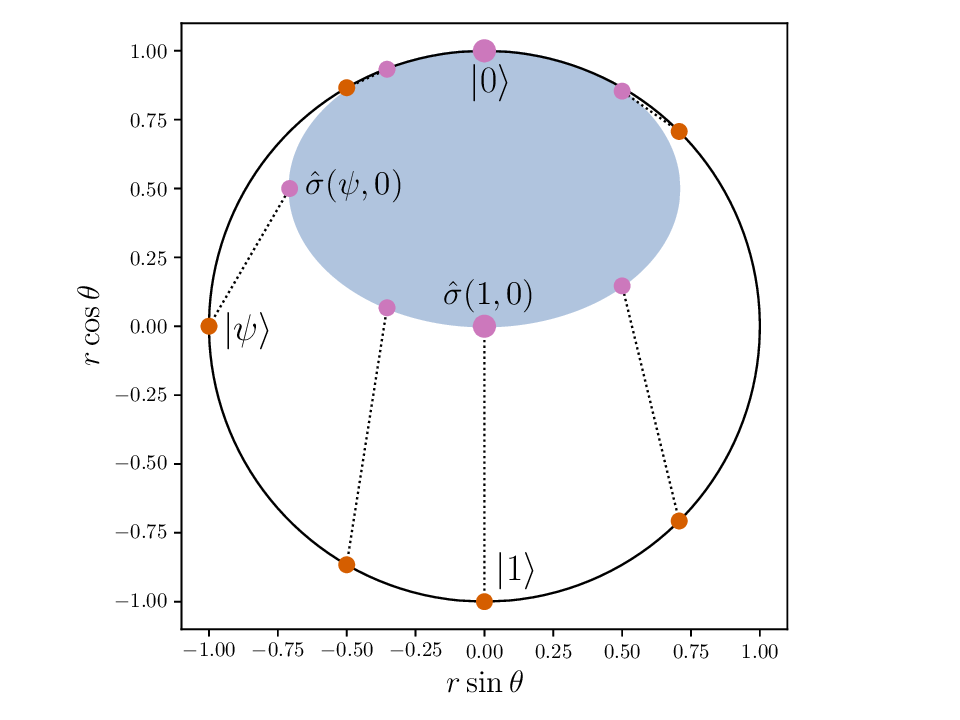}}
    \caption{\small
        Representation of $\mathcal{D}^{1}$ and $\mathcal{D}^{1}_{+}$ (blue area)
        over a slice of the Bloch sphere.
        North pole is $\ket{0}$, South pole is $\ket{1}$.
        Pure states $\ket{\psi_{i}}$ (orange points) are connected (dotted lines) to beam-splitter states $\hat{\sigma}(\psi_i,0)$ (purple points).
    }
    \label{fig:projviaBS}
\end{figure}

As a first example, we focus on the set of WPS $\mathcal{D}^1_{+}$.
In Fig.~\ref{fig:bloch_sphere}, we illustrate the effect of the Vertigo map on states in (a slice of) the Bloch sphere. Starting from a Wigner-positive neighborhood of the vacuum state, we observe that the Vertigo map generates all the extreme points of $\mathcal{D}^{1}_{+}$. Namely, defining $\ket{\psi_{\delta, \phi}}\coloneqq\sqrt{1-\delta}\ket{0}+e^{i\phi}\sqrt{\delta}\ket{1}$ for $\delta\ge0$ and $\phi\in[0,2\pi)$, any extreme point of $\mathcal{D}^{1}_{+}$ can be reached via the Vertigo map starting from some beam-splitter state $\hat\sigma(\psi_{\delta, \phi},0)$ with $\delta\le\varepsilon$, for any $\varepsilon>0$.

Another characterization of the set $\mathrm{Extr}(\mathcal{D}^{1}_{+})$ can be obtained. As can be seen from Fig.~\ref{fig:projviaBS}, each extreme point is of the form $\hat{\sigma}(\psi,0)$ where $\ket{\psi}\in\mathcal{H}^1$.
Now, recall that the Wigner function of $\hat{\sigma}(\psi,0)$ is equal to a rescaled version of the Husimi function of $\ket{\psi}$, and has thus exactly one zero (or no zero if $\ket{\psi}=\ket{0}$, in which case $\hat{\sigma}(\psi,0)=\ket{0}\!\bra{0}$).
In conclusion, the extreme points of $\mathcal{D}^{1}_{+}$ are in one-to-one correspondence with the location of their zero (and vacuum is the only extreme state with no zero).
Note that points in $\mathcal{D}^{1}_{+}$ whose Wigner function vanishes at a single point are necessarily extreme points.

\begin{figure*}[t]
    \centering
    \includegraphics[width=0.6\linewidth]{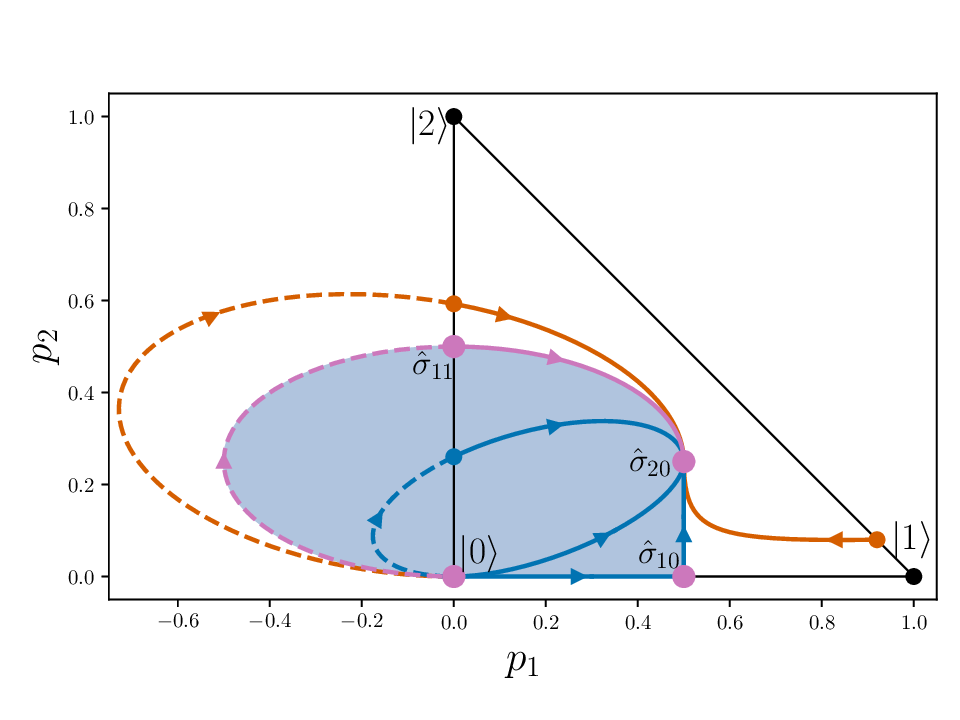}
    \caption{\small
    We consider the set of phase-invariant Fock-bounded quasi-states $\mathcal{A}^{2}_{\odot}$, parametrized as $\hat{A}=(1-p_{1}-p_{2})\ket{0}\!\bra{0}+p_{1}\ket{1}\!\bra{1}+p_{2}\ket{2}\!\bra{2}$. The blue-shaded region is the Wigner-positive subset $\mathcal{A}^{2}_{\oplus}$, and the interior of the triangle corresponds to the density operators $\mathcal{D}^{2}_{\odot}$. Purple dots denote the beam-splitter states $\hat{\sigma}_{mn}\vcentcolon=\hat{\sigma}(m,n)$. The binomial states $\hat{\sigma}_{20}$, $\hat{\sigma}_{10}$, and $\hat{\sigma}_{00}=\ket{0}\!\bra{0}$ are fixed points of $\mathcal{V}^{\mathrm{norm}}_{t}$.
    Vertigo trajectories are shown as curves with arrows indicating their direction: blue curves correspond to Wigner-positive (quasi-)states, orange curves to Wigner-negative (quasi-)states, and purple curves to beam-splitter states. Solid lines denote states; dashed lines denote quasi-states.
    Observe that the Vertigo trajectory of $\hat{\sigma}_{11}$ (purple curve) coincides with the boundary of $\mathcal{A}^{2}_{\oplus}$.
    }
    \label{fig:mixtures_012}
\end{figure*}

In addition, we provide still another characterization of the set $\mathrm{Extr}(\mathcal{D}^{1}_{+})$ as follows.
Start from the extreme point $\hat{\sigma}(1,0)$, whose Wigner function has a single zero at $\alpha=0$.
Then, displace it with the map $\bar{D}^{\mathrm{norm}}_{\beta}$ to bring its zero at $\alpha=\beta$.
This construction allows us to create all the extreme points of $\mathcal{D}^{1}_{+}$ that have a single zero.
Finally, vacuum can be obtained as the limiting case $\bar{D}^{\mathrm{norm}}_{\beta}[\hat{\sigma}(1,0)]$ for $\beta\to\infty$ (see Appendix~\ref{apd:fock-bounded_disp}).
We thus obtain three equivalent characterizations of the extreme points of $\mathcal{D}^{1}_{+}$:
\begin{align}
    \mathrm{Extr}\left(
    \mathcal{D}^{1}_{+}\right)
    = & \big\lbrace\mathcal V^\mathrm{norm}_t[\hat{\sigma}(\psi_{\delta, \phi},0)]\big\rbrace_{\substack{0\leq\delta\leq\varepsilon \\[0.2em]\phi\in[0,2\pi)\\[0.2em]t\geq 1}}\\
    = &
    \big\lbrace\hat{\sigma}(\psi,0)\big\rbrace_{\ket{\psi}\in\mathcal{H}^{1}}
    \label{eq:D1+_beam_splitter_states}
    \\[1em]
    = &
    \big\lbrace\bar{D}^{\mathrm{norm}}_{\beta}[\hat{\sigma}(1,0)]\big\rbrace_{\beta\in\bar{\mathbb{C}}},
\end{align}
where the first equation includes $t\to\infty$ and holds for any choice of $\varepsilon>0$, and where $\bar{\mathbb{C}}=\mathbb{C}\cup\lbrace\infty\rbrace$ is the extended complex plane.
The jump to the extended complex plane is needed in order to ensure that $\ket{0}\!\bra{0}$ is included.

The parametrization \eqref{eq:D1+_beam_splitter_states} of $\mathrm{Extr}(\mathcal{D}^{1}_{+})$ as the set of beam-splitter
states $\hat{\sigma}(\psi,0)=\mathcal{E}_{1/2}[\ket{\psi}\!\bra{\psi}]$ echoes a similar construction presented in our companion
paper~\cite{mathpaper}. There, we introduce a map $F$ sending the boundary of $\mathcal{D}$ onto the boundary of
$\mathcal{D}_{+}$ by classical mixing with the vacuum, and exhibiting behaviour analogous to that of $\mathcal{E}_{1/2}$. Although, in general,
$\mathcal{E}_{1/2}[\ket{\psi}\!\bra{\psi}] \neq F[\ket{\psi}\!\bra{\psi}]$, both maps send $\mathrm{Extr}(\mathcal{D}^{1})$ into
$\mathrm{Extr}(\mathcal{D}^{1}_{+})$.

As a second example, consider now the phase-invariant set $\mathcal{D}^{2}_{\oplus}$, illustrated on Fig.~\ref{fig:mixtures_012}. In agreement with the results of \cite{Van_Herstraeten2021-nj},
its extreme points are the 3 binomial beam-splitter states ($\hat{\sigma}(0,0)$, $\hat{\sigma}(1,0)$, $\hat{\sigma}(2,0)$) and the Vertigo trajectory of the beam-splitter state $\hat{\sigma}(1,1)$.
They can be parameterized as follows:
\begin{align}
    \mathrm{Extr}(\mathcal{D}^{2}_{\oplus})
    =
    \Big\lbrace
    \mathcal{V}^{\mathrm{norm}}_{t}
    \left[
        \hat{\sigma}(m,n)
        \right]
    \Big\rbrace_{\substack{m+n\leq 2 \\[0.2em]t\geq 1}}.
    \label{eq:description_D2oplus_vertigo}
\end{align}
Note that this parameterization is slightly redundant since $\hat{\sigma}(n,0)$ is invariant under the Vertigo map.
Furthermore, the the extreme points of $\mathcal{A}^{2}_{\oplus}$ can be parameterized in a very similar fashion as follows:
\begin{align}
    \mathrm{Extr}(\mathcal{A}^{2}_{\oplus})
    =
    \Big\lbrace
    \mathcal{V}^{\mathrm{norm}}_{t}
    \left[
        \hat{\sigma}(m,n)
        \right]
    \Big\rbrace_{\substack{m+n\leq 2 \\[0.2em]t> 0}},
    \label{eq:description_A2oplus_vertigo}
\end{align}
where the only difference with respect to Eq.~\eqref{eq:description_D2oplus_vertigo} is that the parameter $t$ can take any positive value.

From numerical investigations, we can further analyze the structure of the sets $\mathcal{D}^{3}_{\oplus}$ and $\mathcal{A}^{3}_{\oplus}$.
We observe a qualitative difference between the Vertigo trajectories of $\hat{\sigma}(1,1)$ and $\hat{\sigma}(2,1)$: whereas $\mathcal{V}_{t}[\hat{\sigma}(1,1)]$ exits the PSD cone as soon as $t<1$, the backwards-evolved state $\mathcal{V}^{\mathrm{norm}}_{t}[\hat{\sigma}(2,1)]$ remains physical for some $t<1$.
This difference stems from the fact that $\hat{\sigma}(2,1)$ is full-rank in $\mathcal{D}^3$, whereas $\hat{\sigma}(1,1)$ is not full-rank in $\mathcal{D}^{2}$.
As a result, no analog of Eq.~\eqref{eq:description_D2oplus_vertigo} appears to hold for $\mathrm{Extr}(\mathcal{D}^{3}_{\oplus})$.
In contrast, the case of $\mathrm{Extr}(\mathcal{A}^{3}_{\oplus})$ is more tractable, and we observe that it admits the following explicit characterization:
\begin{align}
    \mathrm{Extr}(\mathcal{A}^{3}_{\oplus})
    =
    \Big\lbrace
    \mathcal{V}^{\mathrm{norm}}_{t}
    \left[
        \hat{\sigma}(m,n)
        \right]
    \Big\rbrace_{\substack{m+n\leq 3 \\[0.2em]t> 0}},
    \label{eq:description_A3oplus_vertigo}
\end{align}
which incrementally generalizes
Eq.~\eqref{eq:description_A2oplus_vertigo}.

It may be tempting to extrapolate Eqs.~\eqref{eq:description_A2oplus_vertigo} and ~\eqref{eq:description_A3oplus_vertigo} to higher orders, but these are in fact specific features that can only occur when the Fock support is sufficiently low (i.e.\ $n\leq 3$).
Indeed, Theorem~\ref{th:extrAo+n} shows that quasi-states of $\mathrm{Extr}(\mathcal{A}^{n}_{\oplus})$ have Wigner functions with at most $\lfloor n/2\rfloor$ circles of zeros (zeros at the origin are not counted as circles).
For $n\leq 3$, such quasi-states have at most $1$ circle of zeros, so that they can all be connected by a single Vertigo trajectory, as the map $\mathcal{V}_{t}$ simply rescales the radius of the circle.
For $n\geq 4$, however, extreme quasi-states may exhibit two or more circles of zeros.
In this case, they can no longer be interconnected by a single Vertigo trajectory, since $\mathcal{V}_t$ preserves the ratio between the radii of the zeros.
For instance, $\hat{\sigma}(2,2)$ is the only beam-splitter state in $\mathrm{Extr}(\mathcal{A}^{4}_{\oplus})$ with two circles of zeros, and its Vertigo trajectory only connects those quasi-states of $\mathrm{Extr}(\mathcal{A}^{4}_{\oplus})$ that share the same ratio of radii as $\hat{\sigma}(2,2)$.

Nevertheless, the following inclusion relations hold true for any $k\in\mathbb{N}$:
\begin{align}
    \Big\lbrace
    \mathcal{V}^{\mathrm{norm}}_{t}
    \left[
        \hat{\sigma}(m,n)
        \right]
    \Big\rbrace_{\substack{m+n\leq k \\[0.2em]t\geq 1}}
     & \subseteq
    \mathrm{Extr}(\mathcal{D}^{k}_{\oplus})
    ,
    \\[0.8em]
    \Big\lbrace
    \mathcal{V}^{\mathrm{norm}}_{t}
    \left[
        \hat{\sigma}(m,n)
        \right]
    \Big\rbrace_{\substack{m+n\leq k \\[0.2em]t> 0}}
     & \subseteq
    \mathrm{Extr}(\mathcal{A}^{k}_{\oplus})
    .
\end{align}
Remember also that by Lemma~\ref{lemma:extreme-fock-diagonal-quasi-states-are-extreme}, we have $\mathrm{Extr}(\mathcal{A}^{k}_{\oplus})\setminus\mathcal{A}^{k-1}_{\oplus}\subset\mathrm{Extr}(\mathcal{A}^{k}_{+})$.

\section{Conclusion and open problems}
\label{sec:concl}

In this work, building upon the characterization of the topological and geometrical properties of the set of WPS uncovered in a companion paper~\cite{mathpaper}, we have introduced new tools for studying and generating extreme WPS.
Our main contributions are the introduction of the Vertigo map, which maps WPQS to WPS and generates trajectories of extreme WPS, together with a characterization of a large class of extreme WPQS based on polynomial analysis.

Our findings reveal that the family of beam-splitter states plays a particular role in the structure of the set of WPS: beamsplitter states obtained by mixing a Fock state with the vacuum are the attracting points of the Vertigo trajectories, while beamsplitter states obtained by mixing two Fock states are extreme WPS. This leads us to formulate the following open problem: for any two Fock-bounded pure states $\ket\psi,\ket\phi$ is the beamsplitter state $\hat\sigma(\psi,\phi)$ an extreme WPS? We note that the answer is negative when the Fock-bounded condition is removed, as can be seen by considering the beam-splitter state obtained from two squeezed vacuum states of orthogonal quadratures.

As an application of our results, we have obtained a general method for constructing extreme WPS, which yields a complete operational parametrization of the extreme WPS in small dimensions. However,this extreme WPS generation method only leads to Wigner functions featuring concentric ellipses of zeros, and there exists WPS with more complicated zero sets. An interesting avenue for future research is to develop more general parametrizations of extreme WPS valid in higher dimensions, for instance based on the set of zeros their Wigner function \cite{abreu2025inverse,mathpaper}. In particular, while our work highlights the challenges of obtaining a complete operational description of the infinite-dimensional WPS set, it also reveals the presence of a rich convex structure that can be exploited to understand this set.




\section*{Acknowledgements}

U.C., Z.V.H.\ and J.D.\ acknowledge funding from the European Union’s Horizon Europe Framework Programme (EIC Pathfinder Challenge project Veriqub) under Grant Agreement No.~101114899.
N.C.D.\ and J.N.P.\ were partially supported by the Fundação para a Ciência e Tecnologia (Portugal) via the research centre GFM, reference UID/00208/2025.
N.J.C.\ acknowledges support from the F.R.S.–FNRS under project
CHEQS within the Excellence of Science (EOS) program.


\bibliographystyle{ieeetr}
\bibliography{biblio}


\appendix
\onecolumn

\newpage

\section{Vertigo map in state space}
\label{apd:vertigo_in_state_space}

The Vertigo map is originally defined from its action in phase space over Wigner functions.
Here, we derive the corresponding representation in terms of quantum operators and quantum maps.

\subsection{Operator-sum representation}

The Wigner function of the Fock transition operators $\ket{m}\!\bra{n}$ is denoted as $W_{\ket{m}\bra{n}}(\alpha)$ and evaluates to:
\begin{align}
    W_{\ket{m}\bra{n}}(\alpha)
    = &
    \frac{2}{\pi}
    (-1)^{m}
    \sqrt{\frac{m!}{n!}}\
    (2\alpha)^{n-m}
    L^{(n-m)}_m\big(4\abs{\alpha}^2\big)
    \exp\big(-2\abs{\alpha}^2\big)
    \label{eq:wigner_funcion_fock_mn}
\end{align}
for $m\leq n$, otherwise use the relation $W_{\ket{n}\bra{m}}=W^{\ast}_{\ket{m}\bra{n}}$.
In order to derive the operator-sum representation of the Vertigo map, we compute the terms $\mathcal{V}_{t}\big[\ket{m}\!\bra{n}\big]$.
From Def.~\ref{eq:def_vertigo_map_wigner}, we compute the Wigner function of $\mathcal{V}_{t}\big[\ket{m}\!\bra{n}\big]$ to be the following:

\begin{align}
    W_{\mathcal{V}_{t}[\ket{m}\bra{n}]}(\alpha)
    =
    \frac{2}{\pi}
    (-1)^{m}
    \sqrt{\frac{m!}{n!}}\
    \big(2\sqrt{t}\alpha\big)^{n-m}
    L^{(n-m)}_m\big(4t\abs{\alpha}^2\big)
    \exp\big(-2\abs{\alpha}^2\big)
    \label{eq:wigner_vertigo_fock_element}
\end{align}

We then use the relation obtained from \cite[\href{https://dlmf.nist.gov/18.18.E12}{(18.18.12)}]{NIST:DLMF}:
\begin{align}
    \frac{L^{(\alpha)}_{n}(\lambda x)}{L^{(\alpha)}_{n}(0)}
    =
    \sum\limits_{k=0}^{n}
    \binom{n}{k}
    \lambda^{k}
    (1-\lambda)^{n-\lambda}
    \frac{L^{(\alpha)}_{k}(x)}{L^{(\alpha)}_{k}(0)}
\end{align}
together with $L^{(\alpha)}_{n}(0)=\binom{n+\alpha}{n}$ to compute:
\begin{align}
    L^{(\alpha)}_{n}(\lambda x)
    =
    (1-\lambda)^{n}
    \sum\limits_{k=0}^{n}
    \binom{n+\alpha}{n-k}\left(\frac{\lambda}{1-\lambda}\right)^{k}
    L^{(\alpha)}_{k}(x).
    \label{eq:rescaling_laguerre}
\end{align}

Injecting \eqref{eq:rescaling_laguerre} in \eqref{eq:wigner_vertigo_fock_element} and some rewriting yields the relation:
\begin{align}
    W_{\mathcal{V}_{t}[\ket{m}\bra{n}]}(\alpha)
    =
    \sqrt{t}^{n-m}
    (t-1)^{m}
    \sqrt{\frac{m!}{n!}}
    \sum\limits_{k=0}^{m}
    \binom{n}{m-k}
    \left(\frac{t}{t-1}\right)^k
    \sqrt{\frac{(n-m+k)!}{k!}}
    W_{\ket{k}\bra{n-m+k}}(\alpha)
\end{align}
which can equivalently be written in terms of quantum operators:
\begin{align}
    \mathcal{V}_{t}[\ket{m}\!\bra{n}]
     & =
    \sqrt{t}^{n-m}
    (t-1)^{m}
    \sqrt{\frac{m!}{n!}}
    \sum\limits_{k=0}^{m}
    \binom{n}{m-k}
    \left(\frac{t}{t-1}\right)^k
    \sqrt{\frac{(n-m+k)!}{k!}}
    \ket{k}\!\bra{n-m+k}
    \\
     & =
    \sqrt{t}^{n-m}
    (t-1)^{m}
    \sqrt{\frac{m!}{n!}}
    \sum\limits_{k=0}^{m}
    \binom{n}{k}
    \left(\frac{t}{t-1}\right)^{m-k}
    \sqrt{\frac{(n-k)!}{(m-k)!}}
    \ket{m-k}\!\bra{n-k}
\end{align}

Then, using the rule $\hat{a}^{\ell}\ket{n}=\sqrt{n!/(n-\ell)!}\ket{n-\ell}$, we can write after some simplification:
\begin{align}
    \mathcal{V}_{t}[\ket{m}\!\bra{n}]
     & =
    \sqrt{t}^{n-m}
    (t-1)^{m}
    \sum\limits_{k=0}^{m}
    \frac{1}{k!}
    \left(\frac{t}{t-1}\right)^{m-k}
    \hat{a}^{k}
    \ket{m}\!\bra{n}
    \hat{a}^{\dagger k}
    \\&=
    \sum\limits_{k=0}^{m}
    \frac{1}{k!}
    \left(\frac{t}{t-1}\right)^{-k}
    \hat{a}^{k}
    \sqrt{t}^{\, m}
    \ket{m}\!\bra{n}
    \sqrt{t}^{\, n}
    \hat{a}^{\dagger k}
    \\&=
    \sum\limits_{k=0}^{\infty}
    \frac{1}{k!}
    \left(\frac{t-1}{t}\right)^{k}
    \hat{a}^{k}
    \sqrt{t}^{\,\hat{n}}
    \ket{m}\!\bra{n}
    \sqrt{t}^{\,\hat{n}}
    \hat{a}^{\dagger k}
\end{align}

By linearity, we can extend the action of the Vertigo map to any operator as:
\begin{align}
    \mathcal{V}_{t}
    \big[\hat{A}\big]
     & =
    \sum\limits_{k=0}^{\infty}
    \frac{1}{k!}
    \left(\frac{t-1}{t}\right)^{k}
    \,
    \hat{a}^{k}
    \,
    \sqrt{t}^{\,\hat{n}}
    \,
    \hat{A}
    \,
    \sqrt{t}^{\,\hat{n}}
    \,
    \hat{a}^{\dagger k}
    \\
     & =
    \sum\limits_{k=0}^{\infty}
    \frac{(t-1)^k}{k!}
    \
    \sqrt{t}^{\,\hat{n}}
    \,
    \hat{a}^{k}
    \,
    \hat{A}
    \,
    \hat{a}^{\dagger k}
    \,
    \sqrt{t}^{\,\hat{n}}
    \label{eq:vertigo_operator-sum}
\end{align}
where we have used the relation $\hat{a}^{k}\sqrt{t}^{\,\hat{n}}=\sqrt{t}^{\,\hat{n}+k}\hat{a}^{k}$.

\subsection{As a combination of the PLC and the NLA}

Eq.~\eqref{eq:vertigo_operator-sum} looks very similar to the Kraus decomposition of the PLC, which we recall hereafter:
\begin{align}
    \mathcal{E}_{\eta}
    \big[
        \hat{A}
        \big]
     & =
    \sum\limits_{k=0}^{\infty}
    \frac{(1-\eta)^k}{k!}
    \sqrt{\eta}^{\,\hat{n}}
    \ \hat{a}^k
    \ \hat{\rho}\
    \hat{a}^{\dagger k}\
    \sqrt{\eta}^{\,\hat{n}}.
\end{align}

In fact, it is possible to express the action of the Vertigo map $\mathcal{V}_{t}$ as a combination of the PLC $\mathcal{E}_{\eta}$ and NLA $\mathcal{N}_{g}$ defined as $\mathcal{N}_{g}\big[\hat{A}\big]=\sqrt{g}^{\,\hat{n}}\,\hat{A}\,\sqrt{g}^{\,\hat{n}}$.
Let us first observe that the PLC and NLA obey the following commutation rules:
\begin{align}
    \mathcal{E}_{\eta}\circ\mathcal{N}_{g}
    =
    \mathcal{N}_{\frac{\eta g}{1-g+\eta g}}\circ\mathcal{E}_{1-g+\eta g}\ ,
    \qquad\qquad
    \mathcal{N}_{g}\circ\mathcal{E}_{\eta}
    =
    \mathcal{E}_{\frac{\eta g}{1-\eta+\eta g}}\circ\mathcal{N}_{1-\eta+\eta g}\ .
\end{align}
Using this, we are able to find various equivalent representations of the Vertigo map.
Let us mention a few of them:
\begin{align}
    \mathcal{V}_{t}
    \ \ =\ \
    \mathcal{N}_{t}
    \circ
    \mathcal{E}_{1/t}
    \circ
    \mathcal{N}_{t}
    \ \ =\ \
    \mathcal{E}_{2}\circ\mathcal{N}_{t}\circ\mathcal{E}_{1/2}
    \ \ =\ \
    \mathcal{E}_{\frac{t}{2t-1}}
    \circ
    \mathcal{N}_{2t-1}
    \ \ =\ \
    \mathcal{N}_{\frac{t}{2-t}}
    \circ
    \mathcal{E}_{2-t}
\end{align}

\section{Identifying extreme points}
\label{app:proofs}

\paragraph{Lemma~\ref{lemma:extreme_superset}.}
\textit{
    Let $\mathcal{C},\mathcal{C}'$ be convex sets such that $\mathcal{C}'\subseteq\mathcal{C}$.
    Then, $\mathrm{Extr}(\mathcal{C})\cap\mathcal{C}'\subseteq\mathrm{Extr}(\mathcal{C}')$.
}

\begin{proof}
    By contradiction.
    Assume $x\in\mathrm{Extr}(\mathcal{C})$ and $x\in\mathcal{C}'$.
    If $x\not\in\mathrm{Extr}(\mathcal{C}')$, then $\exists y_1,y_2\in\mathcal{C}'\setminus\lbrace x\rbrace$ such that $x=\frac12(y_1+y_2)$.
    Since $y_1,y_2\in\mathcal{C}'\subset\mathcal{C}$, this contradicts $x\in\mathrm{Extr}(\mathcal{C})$.
\end{proof}

\paragraph{Lemma~\ref{lemma:extreme_subset_projector}}.
\textit{
    Let $\mathcal{C}\subseteq\mathcal{D}$ be a convex subset of $\mathcal D$ and $\hat{P}$ a projector.
    Let $\mathcal{C}'=\hat{P}\mathcal{D}\hat{P}\cap\mathcal{C}$.
    Then, $\mathrm{Extr(\mathcal{C'})}\subseteq\mathrm{Extr}(\mathcal{C})$.
}
\begin{proof}
    By contradiction.
    Assume that $\hat{\rho}\in\mathrm{Extr}(\mathcal{C}')$ and that $\hat{\rho}\not\in\mathrm{Extr}(\mathcal{C})$.
    This implies that $\exists\hat{\sigma}_1,\hat{\sigma}_2\in\mathcal{C}\setminus\lbrace\hat{\rho}\rbrace$ such that $\hat{\rho}=\frac12(\hat{\sigma}_1+\hat{\sigma}_2)$.
    We then have $\mathrm{Tr}[\hat{P}\hat{\rho}]=\frac12\mathrm{Tr}[\hat{P}\hat{\sigma}_1]+\frac12\mathrm{Tr}[\hat{P}\hat{\sigma}_2]$.
    Since $\hat{\rho}\in\mathcal{C}'$ and $\hat{\sigma}_1,\hat{\sigma}_2\in\mathcal{D}$, this implies $\mathrm{Tr}[\hat{P}\hat{\rho}]=\mathrm{Tr}[\hat{P}\hat{\sigma}_1]=\mathrm{Tr}[\hat{P}\hat{\sigma}_2]=1$ so that $\hat{\sigma}_1,\hat{\sigma}_2\in\mathcal{C}'$, which contradicts $\hat{\rho}\in\mathrm{Extr}(\mathcal{C}')$.
\end{proof}

\section{Phase-invariant quasi-states}
\label{app:extrAo+n}

\paragraph{Lemma~\ref{lemma:extreme-fock-diagonal-quasi-states-are-extreme}.}
\textit{$\mathrm{Extr}(\mathcal{A}^{n}_{\oplus})\setminus\mathcal{A}_{\oplus}^{n-1}\subset\mathrm{Extr}(\mathcal{A}^{n}_{+})$.}

\begin{proof}
    By contradiction.
    Assume $\hat{a}\in\mathrm{Extr}(\mathcal{A}^{n}_{\oplus})\setminus\mathcal{A}^{n-1}_{\oplus}$ and $\hat{a}\not\in\mathrm{Extr}(\mathcal{A}^{n}_{+})$.
    Then $\exists\hat{b}_1,\hat{b}_2\in\mathcal{A}^{n}_{+}\setminus\lbrace{\hat{a}\rbrace}$ such that $\hat{a}=\smash{\frac12}(\hat{b}_1+\hat{b}_2)$.
    Using the dephasing channel $\mathcal{R}$, we find that $\hat{a}=\smash{\frac12\big(\mathcal{R}(\hat{b}_1)+\mathcal{R}(\hat{b}_2)\big)}$.
    Since $\hat{a}\in\mathrm{Extr}(\mathcal{A}^{n}_{\oplus})$, it must be that $\mathcal{R}(\hat{b}_1)=\mathcal{R}(\hat{b}_2)=\hat{a}$.
    Define the off-diagonal operators $\hat{b}^{\mathrm{ang}}_{1}=\hat{b}_1-\hat{a}$ and $\hat{b}^{\mathrm{ang}}_2=\hat{b}_2-\hat{a}$.
    Since $\hat{b}_1,\hat{b}_2$ are Wigner-non-negative, it must be that:
    \begin{align}
        -W_{\hat{b}^{\mathrm{ang}}_1}(\alpha)
        \leq
        W_{\hat{a}}(\alpha),
        \\
        -W_{\hat{b}^{\mathrm{ang}}_2}(\alpha)
        \leq
        W_{\hat{a}}(\alpha).
    \end{align}
    Moreover, since $\hat{a}=\frac12(\hat{b}_1+\hat{b}_2)$, we have that $\hat{b}^{\mathrm{ang}}_1=-\hat{b}^{\mathrm{ang}}_2$.
    Define $P_{\hat{x}}(\alpha)=W_{\hat{x}}(\alpha)\;\mathrm{exp}(2\abs{\alpha}^2)$ for $\hat{x}=\hat{a},\hat{b}_{1}^{\mathrm{ang}},\hat{b}_{2}^{\mathrm{ang}}$.
    It follows that:
    \begin{align}
        \abs{P_{\hat{b}^{\mathrm{ang}}_{1}}(re^{i\varphi})}
        =
        \abs{P_{\hat{b}^{\mathrm{ang}}_{2}}(re^{i\varphi})}
        \leq
        \abs{P_{\hat{a}}(r)}
    \end{align}
    Consider the above expressions as polynomials in the variable $r$.
    From Theorem~\ref{th:extrAo+n}, all the zeros of $P_{\hat{a}}$ are real.
    Since $\hat{a}\in\mathcal{A}^{n}_{\oplus}\setminus\mathcal{A}^{n-1}_{\oplus}$, it must be that $\deg(P_{\hat{b}^{\mathrm{ang}}_1})=\deg(P_{\hat{b}^{\mathrm{ang}}_2})<\deg(P_{\hat{a}})=2n$.
    As a consequence from Lemma~\ref{lemma:poly_with_same_roots}, both $\smash{P_{\hat{b}^{\mathrm{ang}}_1}}$ and $\smash{P_{\hat{b}^{\mathrm{ang}}_1}}$ must be zero.
    We conclude that $\hat{b}_1=\hat{b}_2=\hat{a}$, which contradicts the original assumption.
\end{proof}

\begin{lemma}
    \label{lemma:poly_with_same_roots}
    Let $P,Q\in\mathbb{R}[X]$ be polynomials.
    Let $P$ have roots $\lbrace\lambda_k\rbrace$ such that
    $\lambda_k\in\mathbb{R}\ \forall k$ (resp. $\lambda_k\in\mathbb{R}_{\geq 0}\ \forall k$).
    If $\exists c>0$ such that $\abs{Q(t)}\leq c\abs{P(t)}$ $\forall t\in\mathbb{R}$ (resp. $\forall t\in\mathbb{R}_{\geq 0}$), then $Q\;\propto\; P$.
    Moreover, if $\deg(Q)<\deg(P)$, then $Q(t)=0$.
\end{lemma}

\begin{proof}
    Let $\mu_k$ be the multiplicity associated to each root $\lambda_k$.
    Observe that $\forall k$ we may write:
    \begin{align}
        \lim_{t\rightarrow\lambda_k} \frac{Q(t)}{(t-\lambda_k)^{\mu_k-1}}=0
    \end{align}
    This implies that the roots of $P$ are also roots of $Q$, with same multiplicity.
    Moreover, in the limit $t\rightarrow\infty$, $\abs{Q(t)}\leq c\abs{P(t)}$ implies that $\deg(Q)\leq\deg(P)$, so that $Q$ has at most the same number of roots (with multiplicity) as $P$.
    Thus, if $\deg(Q)=\deg(P)$, then $P\;\propto\;Q$; if $\deg(Q)<\deg(P)$, then $Q(t)=0$.
\end{proof}

\begin{lemma}[Extreme non-negative polynomials over $\mathbb{R}_{\geq 0}$]
    \label{lemma:extreme_poly_R+}
    Let $\mathcal{P}$ be the set of polynomials $P\in\mathbb{R}[X]$ such that $\int_{0}^{\infty} P(t)\exp(-2t)=1$.
    Let $\mathcal{P}_{+}\subset\mathcal{P}$ be the subset of non-negative polynomials $P(t)\geq 0\ \forall t\in\mathbb{R}_{\geq 0}$.
    Then the following statements are equivalent:
    \begin{itemize}
        \item[(i)]
              $P\in\mathrm{Extr}(\mathcal{P}_+)$
        \item[(ii)]
              $P(t)=c\ t^{k}\;\prod_{i=1}^{\ell}(t-\lambda_i)^{2}$, where $k,\ell\in\mathbb{N}$, $\lambda_i>0\;\forall i$ and $c>0$ is a normalization constant
    \end{itemize}
\end{lemma}
\begin{proof}
    $(i)\Rightarrow(ii)$.
    Let $\mathcal{R}$ be the set of polynomials of the form of $(ii)$.
    We will show that the whole set $\mathcal{P}_{+}$ can be obtained by convex mixing elements of $\mathcal{R}$.
    A normalized polynomial $Q$ belongs to $\mathcal{P}_{+}$ provided its strictly positive roots come with even multiplicity, and its complex roots come in conjugated pairs.
    Observe that polynomials of $\mathcal{R}$ only have non-negative roots, whereas $Q$ may have in addition negative and complex roots.
    To create a negative root $\mu<0$, mix $P(x)$ with $xP(x)$ as follows:
    \begin{align}
        P_{\mathrm{mix}}(x)
         & =
        (1-p)P(x)+p\;xP(x)                    \\
         & =p\left(x+\frac{1-p}{p}\right)P(x)
        \\&\;\propto(x-\mu)P(x)
    \end{align}
    where $p$ is set to $p=1/(1-\mu)$ so that $\mu=(p-1)/p$ (since $\mu$ is negative, we have $0<p<1$).
    Then, to create a pair of conjugated complex roots $\lbrace z,z^{\ast}\rbrace$, mix $P(x)$ with $(x-\mathrm{Re}z)^2 P(x)$ as follows:
    \begin{align}
        P_{\mathrm{mix}}(x)
         & =
        (1-p)P(x)+p(x-\mathrm{Re}z)^2P(x)
        \\&=
        p\left((x-\mathrm{Re}z)^2+\frac{1-p}{p}\right)P(x)
        \\&=
        p\Big((x-\mathrm{Re}z)^2+(\mathrm{Im}z)^2\Big)P(x)
        \\&\;\propto\;
        (x-z)(x-z^{\ast})P(x)
    \end{align}
    where we have set $p=1/(1+(\mathrm{Im}z)^2)$ so that $(1-p)/p=(\mathrm{Im}z)^2$.
    Using iterative repetitions of this two methods, any number of negative or complex roots can be created from mixtures of polynomial of $\mathcal{R}$, showing that $\mathcal{P}$ belongs to the convex hull of $\mathcal{R}$.
    Since $\mathcal{R}\subset\mathcal{P}_{+}$, this implies that $\mathrm{Extr}(\mathcal{P}_{+})\subseteq\mathcal{R}$.
    \newline$(ii)\Rightarrow(i)$.
    We now show by contradiction that each polynomial of $\mathcal{R}$ is an extreme polynomial of $\mathcal{P}_{+}$.
    Let $R\in\mathcal{R}$ and assume that $\exists Q_1,Q_2\in\mathcal{P}_{+}\setminus\lbrace R\rbrace$ such that $R=\frac12 (Q_{1}+Q_{2})$.
    Since $Q_2(t)\geq 0\ \forall t\in\mathbb{R}_{+}$, we have that $R(t)\geq \frac12 Q_1(t)\ \forall t\geq\mathbb{R}_{+}$.
    From Lemma~\ref{lemma:poly_with_same_roots}, this implies that $Q_1\propto R$; since both are normalized we have $Q_1=R$, which contradicts the original assumption.
    We have thus shown that $\mathcal{R}\subseteq\mathrm{Extr}(\mathcal{P}_{+})$.
    Since it also holds that $\mathrm{Extr}(\mathcal{P}_{+})\subseteq\mathcal{R}$, we conclude that $\mathrm{Extr}(\mathcal{P}_{+})=\mathcal{R}$.
\end{proof}

Note the inclusion chain $\mathrm{Extr}(\mathcal{A}^{n}_{\oplus})\subset\mathrm{Extr}(\mathcal{A}^{n+1}_{\oplus})$, whereas $\mathrm{Extr}(\mathcal{A}^{n}_{+})\not\subset\mathrm{Extr}(\mathcal{A}^{n+1}_+)$.
As an example, let us show that although $\hat{\sigma}(1,0)\in\mathrm{Extr}(\mathcal{A}^{1}_{+})$ (cf. Eq.~\eqref{eq:wigner_function_bs_state} and Theorem~\ref{th:extrAo+n}), it is not an extreme point of $\mathcal{A}^{2}_{+}$, i.e. $\hat{\sigma}(1,0)\not\in\mathrm{Extr}(\mathcal{A}^{2}_{+})$.
Indeed, let us consider the set of polyanalytic polynomials $\bm{P}_{z}$ with $z\in\mathbb{C}$ of this form:
\begin{align}
    \bm{P}_{z}
    =
    \begin{pmatrix}
        0                 & 0 & \tfrac12 z \\
        0                 & 1 & 0          \\
        \tfrac12 z^{\ast} & 0 & 0
    \end{pmatrix}.
\end{align}
We compute $\bm{\alpha}^{\dagger}\bm{P}_{z}\bm{\alpha}=\abs{\alpha}^{2}+\mathrm{Re}(z\alpha^{2})$, which is non-negative $\forall\alpha\in\mathbb{C}$ as soon as $\abs{z}\leq 1$.
Thus, for each $z\in\mathbb{C}$ such that $\abs{z}\leq 1$, the monomial matrix $\bm{P}_{z}$ is associated with a Wigner-positive quasi-state of $\mathcal{A}^{2}_{+}$.
Now, observe that:
\begin{align}
    \frac{1}{2\pi}
    \int_{0}^{2\pi}
    \bm{P}_{e^{i\theta}}\
    \mathrm{d}\theta
    =
    \begin{pmatrix}
        0 & 0 & 0 \\
        0 & 1 & 0 \\
        0 & 0 & 0
    \end{pmatrix}
    =
    \bm{P}_{\hat{\sigma}(1,0)}
\end{align}
where $\bm{P}_{\hat{\sigma}(1,0)}$ is precisely the monomial matrix of $\hat{\sigma}(1,0)$.
Since $\hat{\sigma}(1,0)$ can be obtained as a convex mixture of quasi-states of $\mathcal{A}^{2}_{+}$, $\hat{\sigma}(1,0)$ is not an extreme point of $\mathcal{A}^{2}_{+}$.

\paragraph{Theorem~\ref{th:extrAo+n}}
(Extreme points of $\mathcal{A}^{n}_{\oplus}$)
The following propositions are equivalent:
\begin{itemize}
    \item[$(i)$]\
          $\hat{A}\in\mathrm{Extr}(\mathcal{A}^{n}_{\oplus})$.
    \item[$(ii)$]\ $W_{\hat{A}}\big(\alpha\big)=P\big(\abs{\alpha}^2\big)\;\mathrm{exp}(-2\abs{\alpha}^2)$ where $P$ is a polynomial $P(t)=c\ t^{k}\prod_{i=1}^{\ell}(t-\lambda_i)^{2}$ such that $\lambda_i>0\ \forall i$, $k+2\ell\leq n$ and $c$ is a positive normalization constant.
    \item[$(iii)$]\
          $\hat{A}=(-1)^{m}\sum_{j=0}^{m}(-1)^{j}j!\, e_{m-j}(\bm{\mu})\,\hat{\sigma}(j,0)$, where the vector $\bm{\mu}$ has $m\leq n$ non-negative components, such that strictly positive components come by pairs.
\end{itemize}
\begin{proof}
    $(i)\Leftrightarrow(ii)$.
    The set $\mathrm{Extr}(\mathcal{A}^{n}_{\oplus})$ can be identified as the set of quasi-states $\hat{A}$ with Wigner function $W_{\hat{A}}(\alpha)=P(\abs{\alpha}^2)\mathrm{exp}(-2\abs{\alpha}^2)$ where $P(t)$ is an extreme polynomial over $\mathbb{R}_{\geq 0}$.
    Such extreme polynomials are characterized by Lemma~\ref{lemma:extreme_poly_R+}, which give us the expression $(ii)$.\newline
    $(ii)\Leftrightarrow(iii)$.
    Define the vector $\bm{\mu}\in\mathbb{R}_{\geq 0}^{k+2\ell}$ as follows:
    \begin{align}
        \bm{\mu}
        =
        \big(
        \underbrace{0,...,0}_{k\text{ entries}},
        \underbrace{\lambda_1,\lambda_1,\lambda_2,\lambda_2,...,\lambda_\ell,\lambda_{\ell}}_{2\ell\text{ entries}}
        \big)
    \end{align}
    The order of the components of $\bm{\lambda}$ does not matter, the only constraints is that all the entries are non-negative, and that positive entries come by pairs.
    \begin{align}
        P(t)
        \  & =\
        c\
        t^{k}
        \prod\limits_{i=1}^{\ell}
        (t-\lambda_{i})^{2}
        \ =\
        c\prod\limits_{i=1}^{2\ell+k}
        (t-\mu_{i})
        \       \\&=\
        c
        \sum\limits_{j=0}^{2\ell+k}
        t^{2\ell+k-j}
        (-1)^j
        e_{j}(\bm{\mu})
        \ =\
        c
        \sum\limits_{j=0}^{2\ell+k}
        t^{j}
        (-1)^{k+j}
        e_{2\ell+k-j}(\bm{\mu})
    \end{align}
    where we have introduced elementary symmetric polynomials defined as:
    \begin{align}
        e_j(\bm{\mu})
        =
        \sum\limits_{1\leq k_1\leq k_2\leq...\leq k_j\leq n}
        \mu_{k_1}\mu_{k_2}\cdots\mu_{k_j}.
    \end{align}
    Now, use the fact that binomial states $\hat{\sigma}(n,0)$ have a polynomial equal to $P(t)=\frac{1}{n!}t^n$, and we find:
    \begin{align}
        \hat{A}
        =
        (-1)^{k}
        \sum\limits_{j=0}^{2\ell+k}
        (-1)^jj!\
        e_{2\ell+k-j}(\bm{\mu})\
        \hat{\sigma}(j,0).
    \end{align}
    Setting $m=2\ell+k$ concludes the proof with the converse direction following from the same derivation.
\end{proof}

\section{Vertigo trajectories}
\label{app:vertigo_traj}

\paragraph{Lemma~~\ref{lemma:vertigo_extreme_trajectories}}(Extreme trajectories under $\mathcal{V}$).
\textit{
If $\hat{A}\in\mathrm{Extr}(\mathcal{A}^{n}_{+})$, then $\mathcal{V}^{\mathrm{norm}}_t[\hat{A}]\in\mathrm{Extr}(\mathcal{A}^{n}_{+})$ for all $t>0$.
}

\begin{proof}
    By contradiction.
    Assume $\mathcal{A}\in\mathrm{Extr}(\mathcal{A}^{n}_{+})$ and $\mathcal{V}^{\mathrm{norm}}_{t}[\hat{A}]\not\in\mathrm{Extr}(\mathcal{A}^{n}_{+})$ for some $t$.
    This means that $\exists\hat{B}_1,\hat{B}_2\in\mathcal{A}^{n}_{+}$ such that $\mathcal{V}_{t}[\hat{A}]\propto\frac12(\hat{B}_{1}+\hat{B}_{2})$.
    Equivalently, we have $\hat{A}\propto\frac12(\mathcal{V}_{1/t}[\hat{B}_1]+\mathcal{V}_{1/t}[\hat{B}_2])$.
    Since $\mathcal{V}$ preserves the Fock-support and Wigner-positivity, this contradicts $\hat{A}\in\mathrm{Extr}(\mathcal{A}^{n}_{+})$.
\end{proof}

\paragraph{Lemma~\ref{lemma:vertigo_quasi-states_to_states}.}
($\mathcal{V}$ maps quasi-states to states).
\textit{
Let $\hat{A}$ be a quasi-state in $\mathcal{A}^{n}$ such that $\bra{n}\hat{A}\ket{n} \neq 0$.
Then there exists $t_0 > 0$ such that $\mathcal{V}^{\mathrm{norm}}_t[\hat{A}] \in \mathcal{D}^{n}$ for all $t \geq t_0$.
}
\begin{proof}
    We first show that any quasi-state $\hat{A} \in \mathcal{A}^{n}$ with $\bra{n}\hat{A}\ket{n}\neq 0$ converges to the binomial state $\hat{\sigma}(n,0)$ under the normalized Vertigo map $\mathcal{V}_t^{\mathrm{norm}}$.
    By Eq.~\eqref{eq:wigner_funcion_fock_mn}, the monomial matrix $\bm{P}$ associated with $\hat{A}$ satisfies $P_{nn} = \frac{2}{\pi}\frac{1}{n!}\,4^n\,\bra{n}\hat{A}\ket{n}$.
    Under the unnormalized Vertigo map, the monomial coefficients transform as $P_{k\ell}\to \sqrt{t}^{\,k+\ell} P_{k\ell}$.
    Thus, for large $t$, the dominant term in the matrix becomes $P_{nn}$, while all others entries scale as lower-order powers of $t$ (since $k+\ell < 2n$) and become negligible.
    Observe now that the binomial state $\hat{\sigma}(n,0)$ has a moment matrix with only one non-zero element, i.e.\ $P_{nn}=\frac{2}{\pi}\frac{1}{n!}4^{n}$ (see Eq.~\eqref{eq:wigner_function_binomial_state}).
    Hence, $\lim_{t\to\infty} \mathcal{V}_t[\hat{A}] \propto \bra{n}\hat{A}\ket{n}\, \hat{\sigma}(n,0)$, and after normalization, $\lim_{t \to \infty} \mathcal{V}_t^{\mathrm{norm}}[\hat{A}] = \hat{\sigma}(n,0)$.

    Now, note that $\hat{\sigma}(n,0) = \frac{1}{2^n} \sum_{k=0}^n \binom{n}{k} \ket{k}\bra{k}$ is a full-rank state in $\mathcal{D}^{n}$, hence it belongs to the interior: $\hat{\sigma}(n,0) \in \mathrm{int}(\mathcal{D}^{n})$.
    By continuity of $\mathcal{V}_t^{\mathrm{norm}}$, we deduce that for $t$ sufficiently large, say $t \geq t_0$, the image $\mathcal{V}_t^{\mathrm{norm}}[\hat{A}]$ remains within $\mathcal{D}^{n}$.
    In addition, since $\mathcal{V}_t^{\mathrm{norm}}$ preserves positive semi-definiteness for $t \geq 1$, we can ensure that the trajectory remains in $\mathcal{D}^{n}$ as soon as it enters it.
\end{proof}

\paragraph{Additional lemmas}

\begin{lemma}\label{lem:eig-op}
    Let $\hat{A}$ be a hermitian operator with bounded Fock support. The following propositions are equivalent:
    \begin{enumerate}[label=(\roman*)]
        \item $\hat{A}$ is an eigenstate of $\mathcal V_t$ for some $t>1$
        \item $\hat{A}$ is an eigenstate of $\mathcal V_t$ for all $t\ge1$
        \item The Wigner function of $\hat{A}$ is the product of a real homogeneous polynomial with the Gaussian function $\alpha\mapsto e^{-2\vert\alpha\vert^2}$. The corresponding eigenvalue is $\sqrt t^m$ where $m$ is the degree-sum of the polynomial part of $W_{\hat{A}}$.
    \end{enumerate}
\end{lemma}

\begin{proof}
    Let $\lambda\in\mathbb R$ and let $\hat{A}$ be a hermitian operator with bounded Fock support such that $\mathcal V_t[\hat{A}]=\lambda_t \hat{A}$. $\hat{A}$ has bounded support so its Wigner function is of the form
    \begin{equation}
        W_{\hat{A}}(\alpha):=P_{\hat{A}}(\alpha)e^{-2\vert\alpha\vert^2},
    \end{equation}
    where $P_{\hat{A}}(\alpha):=\sum_{a,b\le n}p_{ab}\alpha^a\alpha^{\ast b}$ is a polynomial. Then, by Eq.~(\ref{eq:def_vertigo_map_wigner}),
    \begin{equation}
        \mathcal{V}_t[W_{\hat{A}}](\alpha):=P_{\hat{A}}(\sqrt s\alpha)e^{-2\vert\alpha\vert^2}.
    \end{equation}
    We thus obtain
    \begin{equation}
        P_A(\sqrt t\alpha)=\lambda_t P_{\hat{A}}(\sqrt t\alpha),
    \end{equation}
    for all $\alpha\in\mathbb C$. This implies $p_{ab}\sqrt t^{a+b}=\lambda_tp_{ab}$ for all $a+b\le n$. For $\hat{A}$ non identically $0$, this shows that there exists $m\le n$ such that $\lambda_t=\sqrt t^m$ and $p_{ab}=0$ for $a+b\neq m$, i.e., $P_{\hat{A}}$ is a homogeneous polynomial of degree $m$. This proves $(i)\Rightarrow(iii)$.

    Reciprocally, any operator whose Wigner function is the product of the Gaussian function $\alpha\mapsto e^{-2\vert\alpha\vert^2}$ with a real homogeneous polynomial of degree-sum $m$ is an eigenstate of $\mathcal V_t$ with eigenvalue $\sqrt t^m$ for all $t\ge1$. This proves $(iii)\Rightarrow(ii)$.

    $(ii)\Rightarrow(i)$ is trivial.
\end{proof}

We now prove a strengthened result for density operators with bounded Fock support:

\begin{lemma}\label{lem:eig-dens-op}
    Let $\hat{\rho}$ be a density operator with bounded Fock support. The following propositions are equivalent:
    \begin{enumerate}[label=(\roman*)]
        \item $\hat{\rho}$ is an eigenstate of $\mathcal V_t$ for some $t>1$
        \item $\hat{\rho}$ is an eigenstate of $\mathcal V_t$ for all $t\ge1$
        \item There exists $m$ such that $W_{\hat{\rho}}(\alpha)\propto\vert\alpha\vert^{2m} e^{-2\vert\alpha\vert^2}$. The corresponding eigenvalue is $t^m$.
        \item There exists $m$ such that $\hat{\rho}=\frac1{2^m}\sum_{k=0}^m\binom mk\ket k\!\bra k$.
    \end{enumerate}
\end{lemma}

\begin{proof}
    We prove $(i)\Rightarrow(iii)$ and $(iii)\Leftrightarrow(iv)$, the other implications being direct.

    $(i)\Rightarrow(iii)$ Let $m$ be the largest Fock number such that $\langle m|\rho|m\rangle$. Since $\hat{\rho}$ is positive semi-definite, we have $\langle m|\rho|n\rangle=0$ for all $n>m$, and we can write $\rho=\sum_{k,l=0}^m\rho_{kl}\ket k\!\bra l$, with $\rho_{mm}\neq0$. The Wigner function of $\hat{\rho}$ is then given by
    \begin{align}
        W_{\hat{\rho}}(\alpha) & =\sum_{k,l=0}^m\rho_{kl}W_{\ket k\!\bra l}(\alpha)                                                                                                      \\
                               & =\sum_{n=0}^m\rho_{nn}W_{\ket n\!\bra n}(\alpha)+2\sum_{0\le k<l\le m}\Re[\rho_{kl}W_{\ket k\!\bra l}(\alpha)]                                          \\
                               & =\left(\sum_{n=0}^m\rho_{nn}L_n(2\vert\alpha\vert^2)+2\sum_{0\le k<l\le m}\Re[\rho_{kl}L_{k,l}(\alpha,\alpha^{\ast})]\right)e^{-2\vert\alpha\vert^{2}}.
    \end{align}
    The first term in between the parentheses is a real-valued polynomial of the form $\propto\rho_{mm}\vert\alpha\vert^{2m}+P(\alpha,\alpha^{\ast})$, with the degree of $P$ smaller than $2m$ and the second term is a a real-valued polynomial of degree smaller than $2m$.
    By Lemma~\ref{lem:eig-op}, the whole term should be an homogeneous polynomial, which shows that the lower degree terms have to cancel each other out. Hence, $W_{\hat\rho}(\alpha)\propto\vert\alpha\vert^{2m} e^{-2\vert\alpha\vert^2}$.

    $(iii)\Leftrightarrow(iv)$ Up to normalization, $\alpha\mapsto\vert\alpha\vert^{2m} e^{-2\vert\alpha\vert^{2}}$ is the $Q$ function of the Fock state $\ket n$, so it is the Wigner function of the state obtained by mixing $\ket{n}$ and $\ket0$ on a balanced beam splitter and tracing out one of the output modes. A quick computation in Fock basis then proves the result.
\end{proof}

We now show that any density operator with bounded Fock support converges to a fixed point under the action of $\mathcal V_t$.
A natural partition of the set of density operators is obtained as:
\begin{equation}
    \mathcal D^{n}=[\mathcal D^n\setminus\mathcal D^{n-1}]\cup[\mathcal D^{n-1}\setminus\mathcal D^{n-2}]\cup\dots\cup[\mathcal D^{1}\setminus\mathcal D^{0}]\cup\mathcal D^{0},
\end{equation}
where $\mathcal D^{0}=\{\ket0\!\bra0\}$ (we use the convention $\mathcal D^{0}=\emptyset$ here after). For each $k\in\{0,\dots,n\}$, the set $\mathcal D^{k}\setminus\mathcal D^{k-1}$ is the set of operators in $\mathcal D^{k}$ such that $\langle k|\rho|k\rangle>0$.

\begin{lemma}
    \label{lem:fixed_point}
    Let $\hat{\rho}$ be a density operator with bounded Fock support. There exists a unique $m$ such that $\hat{\rho}\in\mathcal D^m\setminus\mathcal D^{m-1}$. Moreover, up to normalization, $\mathcal V_t[\hat{\rho}]$ converges to $\frac1{2^m}\sum_{k=0}^m\binom mk\ket k\!\bra k$ when $t$ goes to infinity.
\end{lemma}

\begin{proof}
    By definition of $\mathcal D^k\setminus\mathcal D^{k-1}$, the value of $m$ is given by the highest Fock state in the support of $\hat{\rho}$, i.e., such that $\langle m|\rho|m\rangle>0$. We write $\hat{\rho}=\sum_{k,l=0}^m\rho_{kl}\ket k\!\bra l$, with $\rho_{mm}>0$. As in the previous proof, up to the Gaussian term $\exp(-2\vert\alpha\vert^{2})$, the Wigner function of $\hat{\rho}$ is a sum of a real-valued polynomial of the form $\propto\rho_{mm}\vert\alpha\vert^{2m}$ and a real-valued polynomial of degree smaller than $2m$. Hence, up to the same Gaussian term, the Wigner function of $\mathcal V_t[\hat\rho]$ is a sum of a real-valued polynomial of the form $\propto\rho_{mm}\vert\alpha\vert^{2m}t^{2m}$ and terms of lower order in $t$. Renormalising and letting $t$ go to infinity gives a Wigner function $\propto\vert\alpha\vert^{2m}$, i.e., $\mathcal V_t[\hat\rho]$ converges to $\frac1{2^m}\sum_{k=0}^m\binom mk\ket k\!\bra k$ when $t$ goes to infinity.
\end{proof}

\section{Beam-splitter states}
\label{app:beam-splitter-states}

\paragraph{Theorem~\ref{theorem:beam-splitter-states-are-extreme}.}
The beam-splitter states $\hat{\sigma}(m,n)$ are extreme WPS $\forall m,n\in\mathbb{N}$:
\begin{align}
    \big\lbrace
    \hat{\sigma}(m,n)
    \big\rbrace_{m,n\in\mathbb{N}}
     & \subset\mathrm{Extr}(\mathcal{D}_{+}).
\end{align}

\begin{proof}
    The Wigner function of the beam-splitter state $\hat{\sigma}(m,n)$ is associated to an extreme non-negative polynomial over $\mathbb{R}_{+}$ (see Lemma~\ref{lemma:extreme_poly_R+}).
    It follows that $\hat{\sigma}(m,n)\in\mathrm{Extr}(\mathcal{A}_{\oplus}^{k})$ for any $k\geq m+n$.
    Setting $k=m+n$, we know from Theorem~\ref{lemma:extreme-fock-diagonal-quasi-states-are-extreme} that $\hat{\sigma}(m,n)\in\mathrm{Extr}(\mathcal{A}^{m+n}_{\oplus})\setminus\mathcal{A}^{m+n-1}_{\oplus}\subset\mathrm{Extr}(\mathcal{A}^{m+n}_{+})$.
    Since $\hat{\sigma}(m,n)\in\mathcal{D}^{m+n}_{+}$, it follows from Lemma~\ref{lemma:extreme_superset} that $\hat{\sigma}\in\mathrm{Extr}(\mathcal{A}^{m+n}_{+})\cap\mathcal{D}^{m+n}_{+}\subset\mathrm{Extr}(\mathcal{D}^{m+n}_+)$.
    Finally, applying Lemma~\ref{lemma:extreme_subset_projector}, we conclude that $\hat{\sigma}(m,n)\in\mathrm{Extr}(\mathcal{D}_{+})$.
\end{proof}

\paragraph{Lemma~\ref{lemma:bs_states_fidelity_fock_states}.}
Beam-splitter states maximize the fidelity with Fock states among Wigner-positive states:
\begin{align}
    \hat{\sigma}(n,n)
    \, & =\,
    \arg\max\limits_{\hat{\rho}\in\mathcal{D}^{2n}_{+}}
    \bra{2n}\hat{\rho}\ket{2n}
    \\
    \hat{\sigma}(n+1,n)
    \, & =\,
    \arg\max\limits_{\hat{\rho}\in\mathcal{D}^{2n+1}_{+}}
    \bra{2n+1}\hat{\rho}\ket{2n+1}
\end{align}

\begin{proof}
    Ref.~\cite{Chabaud2021witnessing} provides the expression of the states maximising the fidelity over Fock states.
    In particular, the maximum fidelity with respect to Fock state $\ket{n}$ is given by the expression:
    \begin{align}
        \bra{n}\hat{\rho}_{\mathrm{max}}\ket{n}
        =
        \frac{1}{2^n}
        \binom{n}{\lfloor \frac{n}{2}\rfloor}.
        \label{eq:max_fidelity_fock_state}
    \end{align}
    Let us compare the above expression with $\bra{2n}\hat{\sigma}(n,n)\ket{2n}$ and $\bra{2n+1}\hat{\sigma}(n+1,n)\ket{2n+1}$.
    We can compute this using the expression from \cite{Van_Herstraeten2021-nj}:
    \begin{align}
        \hat{\sigma}(m,n)
        =
        \frac{2^{-m-n}}{m!n!}
        \sum\limits_{z=0}^{m+n}
        z!(m+n-z)!
        \ \ket{z}\!\bra{z}\
        \left(
        \sum\limits_{j=\max(0,z-n)}^{\min(z,m)}
            (-1)^{j}
        \binom{m}{j}
        \binom{n}{z-j}
        \right)^2.
    \end{align}

    This gives:
    \begin{align}
        \bra{2n}\hat{\sigma}(n,n)\ket{2n}
         & =
        \frac{2^{-2n}}{(n!)^2}
        \ (2n)!
        \underbrace{\left(
            \sum\limits_{j=n}^{n}
            (-1)^{j}
            \binom{n}{j}
            \binom{n}{2n-j}
            \right)^2}_{=1}
        =\frac{1}{2^{2n}}\binom{2n}{n}
    \end{align}
    \begin{align}
        \bra{2n+1}\hat{\sigma}(n+1,n)\ket{2n+1}
         & =
        \frac{2^{-2n-1}}{n!(n+1)!}
        (2n+1)!
        \underbrace{
            \left(
            \sum\limits_{j=\max(n+1)}^{\min(n+1)}
            (-1)^{j}
            \binom{n+1}{j}
            \binom{n}{2n+1-j}
            \right)^2
        }_{=1}
        \\&=
        \frac{1}{2^{2n+1}}
        \binom{2n+1}{n}
    \end{align}
    In both case, we recover the expression of Eq.~\eqref{eq:max_fidelity_fock_state}, which concludes the proof.
\end{proof}

\section{Extremality preserving channels}

\paragraph{Lemma~\ref{lemma:extremality_preserving_map_v2}}(Extremality-preserving map)
Let $\mathcal{M}$ be a linear, invertible (not necessarily trace-preserving) map.
Assume:
\begin{enumerate}
    \item $\mathcal{M},\mathcal{M}^{-1}$ preserve positive semi-definiteness,
    \item $\mathcal{M},\mathcal{M}^{-1}$ preserve Wigner positivity.
\end{enumerate}
Define $\mathcal{M}^{\mathrm{norm}}[\hat{\rho}]=\mathcal{M}[\hat{\rho}]/\mathrm{Tr}\big[\mathcal{M}[\hat{\rho}]\big]$.
Then, $\mathcal{M}^{\mathrm{norm}}:\mathrm{Extr}(\mathcal{D}_{+})\rightarrow\mathrm{Extr}(\mathcal{D}_{+})$.

\begin{proof}
    Assume $\hat{\rho}\in\mathrm{Extr}(\mathcal{D}_{+})$.
    Equivalently, $\not\exists\hat{\rho}_1,\hat{\rho}_{2}\in\mathcal{D}_{+}\setminus{\lbrace\hat{\rho}}\rbrace$ such that $\hat{\rho}=\frac12\hat{\rho}_1+\frac12\hat{\rho}_2$.
    We are going to show that $\mathcal{M}[\hat{\rho}]/\mathrm{Tr}\big[\mathrm{\mathcal{M}[\hat{\rho}]}\big]$ is also extreme.
    We use a proof by contradiction.
    Assume $\mathcal{M}[\hat{\rho}]/\mathrm{Tr}\big[\mathrm{\mathcal{M}[\hat{\rho}]}\big]$ is not extreme.
    Then, we can write:
    \begin{align}
        \frac{\mathcal{M}[\hat{\rho}]}
        {\mathrm{Tr}\big[\mathcal{M}[\hat{\rho}]\big]}
        =
        \frac12\hat{\sigma}_1
        +
        \frac12\hat{\sigma}_2
        \qquad\Leftrightarrow\qquad
        \mathcal{M}[\hat{\rho}]
        =
        \frac{\mathrm{Tr}\big[\mathcal{M}[\hat{\rho}]\big]}{2}\hat{\sigma}_1
        +
        \frac{\mathrm{Tr}\big[\mathcal{M}[\hat{\rho}]\big]}{2}\hat{\sigma}_2.
    \end{align}
    We then apply the inverse channel $\mathcal{M}^{-1}$ on both sides:
    \begin{align}
        \hat{\rho}
         & =
        \mathcal{M}^{-1}
        \left[
            \frac{\mathrm{Tr}\big[\mathcal{M}[\hat{\rho}]\big]}{2}\hat{\sigma}_1
            +
            \frac{\mathrm{Tr}\big[\mathcal{M}[\hat{\rho}]\big]}{2}\hat{\sigma}_2
            \right]
        \\&=
        \frac{\mathrm{Tr}\big[\mathcal{M}[\hat{\rho}]\big]}{2}
        \mathcal{M}^{-1}[\hat{\sigma}_1]
        +
        \frac{\mathrm{Tr}\big[\mathcal{M}[\hat{\rho}]\big]}{2}
        \mathcal{M}^{-1}[\hat{\sigma}_2]
        \\&=
        \frac{\mathrm{Tr}\big[\mathcal{M}[\hat{\rho}]\big]}{2}
        \underbrace{\frac{\mathcal{M}^{-1}[\hat{\sigma}_1]}{\mathrm{Tr}\big[\mathcal{M}^{-1}[\hat{\sigma}_1]\big]}}_{\hat{\tau}_1}
        \mathrm{Tr}\big[\mathcal{M}^{-1}[\hat{\sigma}_1]\big]
        +
        \frac{\mathrm{Tr}\big[\mathcal{M}[\hat{\rho}]\big]}{2}
        \underbrace{\frac{\mathcal{M}^{-1}[\hat{\sigma}_2]}{\mathrm{Tr}\big[\mathcal{M}^{-1}[\hat{\sigma}_2]\big]}}_{\hat{\tau}_2}
        \mathrm{Tr}\big[\mathcal{M}^{-1}[\hat{\sigma}_2]\big]
        \\&=
        \frac{1}{2}
        \mathrm{Tr}\big[\mathcal{M}[\hat{\rho}]\big]
        \mathrm{Tr}\big[\mathcal{M}^{-1}[\hat{\sigma}_1]\big]
        \ \hat{\tau}_1
        +
        \frac{1}{2}
        \mathrm{Tr}\big[\mathcal{M}[\hat{\rho}]\big]
        \mathrm{Tr}\big[\mathcal{M}^{-1}[\hat{\sigma}_2]\big]
        \ \hat{\tau}_2
    \end{align}

    By construction, the operators $\hat{\rho},\hat{\tau}_1,\hat{\tau}_2$ are normalized and we have:
    \begin{align}
        \mathrm{Tr}\big[\mathcal{M}[\hat{\rho}]\big]\mathrm{Tr}\big[\mathcal{M}^{-1}[\hat{\sigma}_1]\big]+\mathrm{Tr}\big[\mathcal{M}[\hat{\rho}]\big]\mathrm{Tr}\big[\mathcal{M}^{-1}[\hat{\sigma}_2]\big]=2.
    \end{align}
    Moreover, if $\mathcal{M},\mathcal{M}^{-1}$ preserve Wigner-positivity, we have $\mathrm{Tr}\big[\mathcal{M}[\hat{\rho}]\big]\geq 0,$, $\mathrm{Tr}\big[\mathcal{M}^{-1}[\hat{\sigma}_1]\big]\geq 0,$ and $\mathrm{Tr}\big[\mathcal{M}^{-1}[\hat{\sigma}_2]\big]\geq 0$.
\end{proof}

\subsection{Fock-bounded displacement map}
\label{apd:fock-bounded_disp}

The Fock-bounded displacement map is defined by its action on Wigner functions as follows:
\begin{align}
    W_{\bar{D}_\beta[\hat{\rho}]}(\alpha)
    =\
    W_{\hat{\rho}}(\alpha-\beta)
    \ \exp\big(
    2\abs{\alpha-\beta}^2-2\abs{\alpha}^2
    \big).
    \label{eq:fock_bounded_disp_wigner}
\end{align}

The displacement map acts on the monomial polynomial as $P(\alpha,\alpha^{\ast})=\bm{\alpha}^{\dagger}\bm{P}\bm{\alpha}$ to $P(\alpha,\alpha^{\ast})=(\bm{\alpha}-\bm{\beta})^{\dagger}\bm{P}(\bm{\alpha}-\bm{\beta})$.

In order to find the operator-sum representation of $\bar{D}_{\beta}$, we are first going to find the operator-sum representation of the operator that displaces the Husimi function in the same fashion as Eq.~\eqref{eq:fock_bounded_disp_wigner}, and then transpose it to the Wigner function using the generalized PLC.
Define the Husimi-Fock-bounded map as follows:
\begin{align}
    Q_{\bar{D}^{\mathrm{H}}_{\beta}[\hat{\rho}]}(\alpha)
    =
    Q_{\hat{\rho}}(\alpha-\beta)
    \ \exp\big(
    \abs{\alpha-\beta}^2-\abs{\alpha}^2
    \big)
    \label{eq:fock_bounded_disp_husimi}
\end{align}
Remember that the PLC channel $\mathcal{E}_{1/2}$ maps the Wigner function to the Husimi, and the PLC channel $\mathcal{E}_{2}$ maps the Husimi function to the Wigner:
\begin{align}
    W_{\mathcal{E}_{1/2}[\hat{\rho}]}(\alpha)
    =
    2
    Q_{\hat{\rho}}(\sqrt{2}\alpha),
    \qquad\qquad
    Q_{\mathcal{E}_{2}[\hat{\rho}]}(\alpha)
    =
    \frac{1}{2}
    W_{\hat{\rho}}\left(\frac{\alpha}{\sqrt{2}}\right).
\end{align}
Together with Eq.~\eqref{eq:fock_bounded_disp_husimi}, the two above relations imply that:
\begin{align}
    \bar{D}_{\beta}
    =
    \mathcal{E}_{1/2}
    \circ
    \bar{D}^{\mathrm{H}}_{\sqrt{2}\beta}
    \circ
    \mathcal{E}_{2}.
    \label{eq:fock_bounded_disp_husimi_to_wigner}
\end{align}
Indeed, we find:
\begin{align}
    W_{\mathcal{E}_{1/2}\Big[\bar{D}^{\mathrm{H}}_{\sqrt{2}\beta}
    \big[
    \mathcal{E}_{2}[\hat{\rho}]\big]\Big]}(\alpha)
     & =
    2Q_{\bar{D}^{\mathrm{H}}_{\sqrt{2}\beta}
    \big[
    \mathcal{E}_{2}[\hat{\rho}]\big]}(\sqrt{2}\alpha)
    \\&=
    2Q_{\mathcal{E}_{2}[\hat{\rho}]}(\sqrt{2}\alpha-\sqrt{2}\beta)\exp(\vert\sqrt{2}\alpha-\sqrt{2}\beta\vert^2-\vert\sqrt{2}\alpha\vert^2)
    \\&=
    W_{\hat{\rho}}(\alpha-\beta)\exp(2\vert\alpha-\beta\vert^2-2\vert\alpha\vert^2)
    \\&=
    W_{\bar{D}_{\beta}[\hat{\rho}]}(\alpha)
\end{align}

To compute the operator-sum representation of $\bar{D}^{\mathrm{H}}_{\beta}$, let us first compute its action on the Fock basis.
We first use the fact that:
\begin{align}
    Q_{\ket{m}\bra{n}}(\alpha)
    \ =\
    \frac{\alpha^{\ast m}}{\sqrt{m!}}
    \
    \frac{\alpha^{n}}{\sqrt{n!}}
    \
    \frac{\exp(-\vert\alpha\vert^2)}{\pi}
\end{align}
in order to compute
\begin{align}
    Q_{\bar{D}^{\mathrm{H}}_{\beta}[\ket{m}\bra{n}]}(\alpha)
     & =
    \frac{1}{\sqrt{m!n!}}
    (\alpha^\ast-\beta^\ast)^m
    (\alpha-\beta)^n
    \frac{\exp(-\vert\alpha\vert^2)}{\pi}
    \\&=
    \frac{1}{\sqrt{m!n!}}
    \sum\limits_{p=0}^{m}
    \sum\limits_{q=0}^{n}
    \binom{m}{p}
    \alpha^{\ast p}
    (-\beta^{\ast})^{m-p}
    \binom{n}{q}
    \alpha^{q}
    (-\beta)^{n-q}
    \frac{\exp(-\vert\alpha\vert^2)}{\pi}
    \\&=
    \sum\limits_{p=0}^{m}
    \sum\limits_{q=0}^{n}
    \sqrt{\frac{p!q!}{m!n!}}
    \frac{m!}{p!(m-p)!}
    \frac{n!}{q!(n-q)!}
    (-\beta^{\ast})^{m-p}
    (-\beta)^{n-q}
    \
    \frac{\alpha^{\ast p}}{\sqrt{p!}}
    \frac{\alpha^{q}}{\sqrt{q!}}
    \frac{\exp(-\vert\alpha\vert^2)}{\pi}
\end{align}
We can now identify the terms associated to $\ket{p}\bra{q}$ in the previous expression and write:
\begin{align}
    \bar{D}^{\mathrm{H}}_{\beta}[\ket{m}\bra{n}] & =
    \sum\limits_{p=0}^{m}
    \sum\limits_{q=0}^{n}
    \sqrt{\frac{p!q!}{m!n!}}
    \frac{m!}{p!(m-p)!}
    \frac{n!}{q!(n-q)!}
    (-\beta^{\ast})^{m-p}
    (-\beta)^{n-q}
    \ket{p}\!\bra{q}
    \\&=
    \sum\limits_{p=0}^{m}
    \sum\limits_{q=0}^{n}
    \frac{1}{(m-p)!}
    \frac{1}{(n-q)!}
    \sqrt{\frac{m!n!}{p!q!}}
    (-\beta^{\ast})^{m-p}
    (-\beta)^{n-q}
    \ket{p}\!\bra{q}
    \\&=
    \sum\limits_{p=0}^{m}
    \sum\limits_{q=0}^{n}
    \frac{1}{p!}
    \frac{1}{q!}
    \sqrt{\frac{m!n!}{(m-p)!(n-q)!}}
    (-\beta^{\ast})^{p}
    (-\beta)^{q}
    \ket{m-p}\!\bra{n-q}
    \\&=
    \left(
    \sum\limits_{p=0}^{m}
    \frac{-\beta^{\ast p}}{p!}
    \sqrt{\frac{m!}{(m-p)!}}
    \ket{m-p}
    \right)\left(
    \sum\limits_{q=0}^{n}
    \frac{-\beta^{q}}{q!}
    \bra{n-q}
    \sqrt{\frac{n!}{(n-q)!}}
    \right)
    \\&=
    \left(
    \sum\limits_{p=0}^{\infty}
    \frac{-\beta^{\ast p}}{p!}
    \hat{a}^p
    \right)
    \ket{m}\!\bra{n}
    \left(
    \sum\limits_{q=0}^{\infty}
    \frac{-\beta^{q}}{q!}
    \hat{a}^{\dagger q}
    \right)
    \\&=
    \exp(-\beta^\ast\hat{a})
    \ket{m}\!\bra{n}
    \exp(-\beta\hat{a}^{\dagger})
\end{align}

So the Husimi-Fock-bounded displacement operator is $\bar{D}^\mathrm{H}_{\beta}[\hat{\rho}]=\exp(-\beta^\ast\hat{a})\hat{\rho}\exp(-\beta\hat{a}^{\dagger})$.

At this point, we establish the relation
$\sqrt{\eta}^{\,\hat{n}} \exp(-\beta^\ast \hat{a})
    = \exp(-\beta^\ast \hat{a}/\sqrt{\eta})\,\sqrt{\eta}^{\,\hat{n}}$,
as it will be useful in what follows.
It follows from the simple manipulation:
\begin{align}
    \sqrt{\eta}^{\hat{n}}
    \exp(-\beta^{\ast}\hat{a})
     & =
    \sqrt{\eta}^{\hat{n}}
    \sum\limits_{k=0}^{\infty}
    \frac{(-\beta^{\ast})^k}{k!}
    \hat{a}^{k}
    \\&=
    \sum\limits_{k=0}^{\infty}
    \frac{(-\beta^{\ast})^k}{k!}
    \hat{a}^{k}
    \sqrt{\eta}^{\hat{n}-k}
    \\&=
    \sum\limits_{k=0}^{\infty}
    \frac{(-(\beta^{\ast}/\sqrt{\eta}))^k}{k!}
    \hat{a}^{k}
    \sqrt{\eta}^{\hat{n}}
    \\&=
    \exp\left(-\frac{\beta^{\ast}}{\sqrt{\eta}}\hat{a}\right)
    \sqrt{\eta}^{\hat{n}}
\end{align}

We are now equipped to determine how $\bar{D}^{\mathrm{H}}_{\beta}$ composes with $\mathcal{E}_{\eta}$:

\begin{align}
    \mathcal{E}_{\eta}\circ\bar{D}^{\mathrm{H}}_{\beta}[\hat{\rho}]
    = & \ \mathcal{E}_{\eta}\big[\exp(-\beta^\ast\hat{a})\;\hat{\rho}\;\exp(-\beta\hat{a}^{\dagger})\big]
    \\=&
    \sum\limits_{k=0}^{\infty}
    \frac{(1-\eta)^k}{k!}
    \sqrt{\eta}^{\,\hat{n}}
    \ \hat{a}^k
    \exp(-\beta^\ast\hat{a})\;\hat{\rho}\;
    \exp(-\beta\hat{a}^\dagger)\
    \hat{a}^{\dagger k}\
    \sqrt{\eta}^{\,\hat{n}}
    \\=&
    \sum\limits_{k=0}^{\infty}
    \frac{(1-\eta)^k}{k!}
    \sqrt{\eta}^{\,\hat{n}}
    \
    \exp(-\beta^\ast\hat{a})
    \
    \hat{a}^k
    \;\hat{\rho}\;
    \hat{a}^{\dagger k}
    \
    \exp(-\beta\hat{a}^\dagger)
    \
    \sqrt{\eta}^{\,\hat{n}}
    \\=&
    \sum\limits_{k=0}^{\infty}
    \frac{(1-\eta)^k}{k!}
    \
    \exp(-\beta^\ast\hat{a}/\sqrt{\eta})
    \
    \sqrt{\eta}^{\,\hat{n}}
    \hat{a}^k
    \;\hat{\rho}\;
    \hat{a}^{\dagger k}
    \
    \sqrt{\eta}^{\,\hat{n}}
    \
    \exp(-\beta\hat{a}^\dagger/\sqrt{\eta})
    \\=&
    \exp(-\beta^{\ast}\hat{a}/\sqrt{\eta})
    \
    \mathcal{E}_{\eta}
    [\hat{\rho}]
    \
    \exp(-\beta\hat{a}^{\dagger}/\sqrt{\eta})
    \\=&
    \bar{D}^{\mathrm{H}}_{\beta/\sqrt{\eta}}\circ\mathcal{E}_{\eta}
    [\hat{\rho}]
\end{align}

We can now use Eq.~\eqref{eq:fock_bounded_disp_husimi_to_wigner} to compute $\bar{D}_{\beta}$.
\begin{align}
    \bar{D}_{\beta}
    \ =\
    \mathcal{E}_{2}
    \circ
    \bar{D}^{\mathrm{H}}_{\sqrt{2}\beta}
    \circ\mathcal{E}_{1/2}
    \ =\
    \bar{D}^{\mathrm{H}}_{\beta}
    \circ\mathcal{E}_{2}
    \circ\mathcal{E}_{1/2}
    \ =\
    \bar{D}^{\mathrm{H}}_{\beta}
\end{align}

We can then finally write:
\begin{align}
    \bar{D}_{\beta}\big[\hat{\rho}\big]
    =
    \exp\left(-\beta^{\ast}\hat{a}\right)\;
    \hat{\rho}\;
    \exp\left(-\beta\hat{a}^{\dagger}\right)
\end{align}
The map $\bar{D}_{\beta}$ preserves the sign of the eigenvalues.
We define $\bar{D}_{\beta}^{\mathrm{norm}}[\hat{\rho}]=\bar{D}_{\beta}[\hat{\rho}]/\mathrm{Tr}[\bar{D}_{\beta}[\hat{\rho}]]$.
The Fock-bounded displacement map preserves the Fock support (i.e., $\bar{D}^{\mathrm{norm}}_{\beta}:\mathcal{A}^{n}\rightarrow\mathcal{A}^{n}$), Wigner positivity (i.e., $\bar{D}_{\beta}^{\mathrm{norm}}:\mathcal{A}_{+}\rightarrow\mathcal{A}_{+}$), and positive semi-definiteness (i.e., $\bar{D}_{\beta}^{\mathrm{norm}}:\mathcal{D}\rightarrow\mathcal{D}$).

Let us finally show that in the limiting case $\beta\to\infty$, the map $\bar{D}^{\mathrm{norm}}_{\beta}$ sends Fock-bounded states to vacuum.
We compute:
\begin{align}
    \exp(-\beta^{\ast}\hat{a})
    \ket{n}
    \ =\
    \sum\limits_{k=0}^{\infty}
    \frac{1}{k!}
    \big(-\beta^{\ast}\big)^k
    \hat{a}^{k}
    \ket{n}
    \ =\
    \sum\limits_{k=0}^{n}
    \frac{1}{k!}
    \big(-\beta^{\ast}\big)^k
    \sqrt{\frac{n!}{(n-k)!}}
    \ket{n-k},
\end{align}
from which we find:
\begin{align}
    \forall\ket{\psi}\in\mathcal{H}^{n}:
    \qquad
    \lim\limits_{\beta\rightarrow\infty}
    \exp(-\beta^{\ast}\hat{a})
    \ket{\psi}
    \propto
    \ket{0}
\end{align}
and then by linearity we conclude:
\begin{align}
    \forall\hat{\rho}\in\mathcal{D}^{n}:
    \qquad
    \lim\limits_{\beta\to\infty}
    \bar{D}^{\mathrm{norm}}_{\beta}
    \left[\hat{\rho}\right]
    =\ket{0}\!\bra{0}.
\end{align}

\end{document}